\documentclass{elsarticle}

\usepackage[cmex10]{amsmath}
\usepackage{amssymb}
\usepackage{amsthm }
\usepackage{nicefrac}

\newtheorem{theorem}{Theorem}

\theoremstyle{theorem}

\theoremstyle{definition}

\theoremstyle{plain}

\theoremstyle{plain}
\newtheorem{Remark}{Remark}[]

\usepackage{hyperref}
\usepackage[bottom]{footmisc}
\usepackage{url}
\usepackage{xcolor}
\hypersetup{colorlinks=true}
\usepackage{tabularx}

\usepackage{algorithm}
\usepackage{multirow}
\usepackage{algcompatible}
\usepackage{adjustbox}

\usepackage{epsfig,epsf,epstopdf,graphicx}
\usepackage[caption=false,font=footnotesize]{subfig}

\usepackage{fixltx2e}
\usepackage{afterpage}
\usepackage{float}
\usepackage{stfloats}

\newcommand{\cb}{{\bf{c}}}

\newcommand{\zerovec}{{\bf{0}}}

\newcommand{\Bb}{\mathbf{B}}

\newcommand{\Db}{\mathbf{D}}

\newcommand{\Hb}{\mathbf{H}}
\newcommand{\Ib}{\mathbf{I}}

\newcommand{\Lb}{\mathbf{L}}

\newcommand{\Pb}{\mathbf{P}}

\newcommand{\Ub}{\mathbf{U}}

\newcommand{\vb}{\mathbf{v}}
\newcommand{\Wb}{\mathbf{W}}

\newcommand{\xb}{\mathbf{x}}
\newcommand{\yb}{\mathbf{y}}
\newcommand{\eb}{\mathbf{e}}
\newcommand{\rb}{\mathbf{r}}
\newcommand{\hb}{\mathbf{h}}

\begin{document}
\begin{frontmatter}
\title{Robust Adaptive Generalized Correntropy-based Smoothed Graph Signal Recovery with a Kernel Width Learning}

\author[mymainaddress]{Razieh~Torkamani}
\ead{hut.torkamani@gmail.com}

\author[mysecondaryaddress]{Hadi~Zayyani\corref{mycorrespondingauthor}}
\ead{zayyani@qut.ac.ir}
\cortext[mycorrespondingauthor]{Corresponding author}

\author[mythirdaddress]{Farokh~Marvasti}
\ead{marvasti@sharif.edu}

\address[mymainaddress]{Qom University of Technology (QUT), Qom, Iran}
\address[mysecondaryaddress]{Qom University of Technology (QUT), Qom, Iran}
\address[mythirdaddress]{Sharif University of Technology, Tehran, Iran}

\begin{abstract}
This paper proposes a robust adaptive algorithm for smooth graph signal recovery which is based on generalized correntropy. A proper cost function is defined, which takes the smoothness and generalized correntropy into account. The generalized correntropy used in this paper employs the generalized Gaussian density (GGD) function as the kernel. The proposed adaptive algorithm is derived and a kernel width learning-based version of the algorithm is suggested. The simulation results confirm the performance of the proposed algorithm for learning the kernel-width to the fixed correntropy kernel version of the algorithm. Moreover, some theoretical analysis of the proposed algorithm are provided. In this regard, firstly, the convexity analysis of the cost function is discussed. Secondly, the uniform stability of the algorithm is investigated. Thirdly, the mean convergence analysis is also added. Finally, the computational complexity analysis of the algorithm is incorporated. In addition, some synthetic and real-world experiments show the efficiency of the proposed algorithm in comparison to some other adaptive algorithms in the literature of adaptive graph signal recovery.
\end{abstract}

\begin{keyword}
Graph signal recovery\sep Generalized Correntropy\sep Robust\sep Non-Gaussian noise.
\end{keyword}

\end{frontmatter}

\section{Introduction}
\label{sec:Intro}
Graph Signal Processing (GSP) which deals with the processing of signals defined irregularly on a graph is nowadays has gained much interest. This filed of signal processing has many applications in biological, social, Internet of Things networks and image processing \cite{Shum13}, \cite{GSP18}. Many problems arises in GSP which Graph Signal Recovery (GSR) \cite{Loren18book}-\cite{Ahmad20}, graph sampling \cite{Nguy20}-\cite{Yang21}, topology learning \cite{Sega17}, and graph signal filtering \cite{Isuf17} to name a few.

In GSR, the entire graph signal should be recovered from a sampled, noisy subset of graph signal. To address this issue, various approaches have recently been developed for graph signal reconstruction, based on the inherent structural information exists in the graph signal. These algorithms are divided into two main categories which are non-adaptive \cite{Chen15}-\cite{Rami21} and adaptive methods \cite{Loren16}-\cite{Ahmad20}. In non-adaptive methods, the graph signal is reconstructed based on some properties such as bandlimitedness of graph signal in a single trial of the graph signal. On the other hand, the adaptive methods recover the graph signal based on streaming sequence of sampled graph signal and in an online and adaptive manner. Adaptive GSR methods have the main benefits of lower computational complexity, be online, and have the ability to adapt itself to nonstationary conditions of graph signal.

The pioneering work \cite{Loren16} generalized the Least Mean Square (LMS) adaptive filtering algorithm to the graph domain. In \cite{Loren17}, a distributed adaptive learning graph signal recovery algorithm is proposed. Moreover, \cite{Loren18} suggested a Recursive Least Square (RLS) algorithm for adaptive GSR. Also, in \cite{LorenCeci18}, LMS strategies are used for online time-varying graph signal recovery in the dynamic graphs. In addition, \cite{Shen19} proposed an online kernel-based graph signal recovery which has scalability and privacy property. In \cite{Spel20}, a Normalized LMS (NLMS) graph signal recovery algorithm is developed which converges faster than LMS and has less complexity than RLS algorithm. Besides, \cite{Loren19}  proposed a collaborative and distributed algorithm based on proximal gradient optimization approach for graph signal learning, tracking, and anomaly detection based on distributed subspace projections. Moreover, \cite{Ahmad20} proposed two modified LMS algorithms which have faster convergence than LMS while maintain the low computational complexity of LMS.

In this paper, the goal is to devise an adaptive graph signal recovery algorithm in presence of impulsive noise. The Gaussian noise assumption is violated in practical scenarios where there exist a heavy tailed impulsive noise. Most adaptive GSR algorithms in the literature \cite{Loren16}-\cite{Ahmad20} relies on the Gaussian assumption of the noise probability distribution. Hence, their performance are degraded in the impulsive environments. In contrast, \cite{Nguy20} considers impulsive noise assumption and suggest an adaptive Least Mean p'th Power (LMP) algorithm for graph signal recovery in alpha-stable noise. A denoising method for data corrupted simultaneously by impulsive noise and additive Gaussian noise was given in \cite{Tay21}, \cite{Jiang21}. In \cite{Tay21}, wireless sensor network is first modelled using an extended graph and a recursive graph median filter is developed that can be implemented with distributed processing. The filter is applied to the denoising of data that is subjected, simultaneously, to Gaussian noise and impulsive noise. In \cite{Jiang21}, a median filter is used that leverages the joint correlation, for denoising time-varying graph signals. The filters are implemented distributively using only information from immediate neighbours and are therefore suitable for resource limited sensor nodes. Both Gaussian noise and impulsive noise were considered.

In this paper, we adopt the well known concept of correntropy-based adaptive filtering \cite{Liu07}-\cite{Lv21} to combat the impulsive noise in GSR. In particular, propose a robust adaptive GSR algorithm which uses the generalized correntropy \cite{Chen16}. The uniform convergence condition of the proposed algorithm and the convergence of the mean of the algorithm is investigated. The simulation results show the superiority of the correntropy-based algorithm in comparison to others. Moreover, experiments on two impulsive noise models, verify the robustness of the proposed algorithm; and experiments with different types of signals, different sampling rates, and different impulsive noise models, show its adaptivity to the signal and noise conditions.

The organization of the paper is as follows. After introduction, sections ~\ref{sec:smooth} and ~\ref{sec:Gen_Corr} present the preliminaries needed in the paper. Then, section~\ref{sec: proposed} developes the proposed algorithm for adaptive graph signal recovery. Section~\ref{sec: theory} provides some theoretical analysis of the algorithm which consists of convexity analysis, uniform stability analysis, mean performance analysis, and complexity analysis of the proposed algorithm. In Section~\ref{sec: simulation}, the simulation experiments are presented. Finally, conclusions are drawn in Section~\ref{sec: conclusion}.

\section{Smoothed graph signal recovery model}
\label{sec:smooth}
%\label{sec:ProblemForm}
Let $\mathcal{G}=\{\mathcal{V},\mathcal{E}\}$ defines an undirected graph consisting of a set of $N$ nodes ${\cal V}=\{1,2,....N\}$. A set edges ${\cal E}=\{W_{ij}\}_{i,j\in{\cal V}}$ is defined such that $W_{ij}>0$ if there is a link between nodes $i$ and $W_{ij}=0$ otherwise. The adjacency matrix $\Wb$ is an $N \times N$ symmetric matrix and consists the collection of all the weights $W_{ij} (i,j=1,2,...,N)$. Suppose $D_i=\sum_{j=1}^{N}{W_{ij}}$ denote the degree of node $i$; then the degree matrix of a graph is defined as a diagonal matrix $\Db=\mathrm{diag}(D_i)$. The Laplacian matrix is also defined as $\Lb =\Db-\Wb$.

A graph signal is shown as $\xb\in {\mathbb R}^{N\times 1}$, where $N$ is the number of nodes. Given the signals of $M$ vertices, the signal at time index $n$ can be written as
\begin{equation}
\label{eq: model}
\yb[n]=\Phi[n]\xb+\vb[n]
\end{equation}
where $\vb$ is the is the $N\times 1$ noise vector which no assumption is made for it throughout this paper, and $\Phi[n]\in\mathbf{R}^{N\times N}$ is the diagonal sampling matrix at time $n$ in which the $i'$th diagonal element is ``one'' if and only if the $i'$th node has been sampled, otherwise this element is ''zero'' and this row takes the value of all-zero.

The aim of GSR is to reconstruct the graph signal $\xb$ from the sampled observations $\yb[n]$. To this end, it is necessary to assume some prior knowledge about the signals $\xb$. On the other hand, most of the real-world graph signals are smooth on the underlying graph, such that the value of neighbor nodes tend to be more analogous than those of distant nodes. In this paper, we exploit the smoothness assumption of deterministic graph signals in the recovery process, which can be characterized by a quadratic form as
\begin{equation}
S(\xb)=\lVert \Lb^{\frac{1}{2}}\xb\rVert=\xb^T\Lb\xb,
\end{equation}
which measures the variations of signal samples over the graph, and can be rewritten as
\begin{equation}
S(\xb)=\sum_{(i,j)\in {\cal E}}{W_{ij}(x_j-x_i)^2},
\end{equation}
where $x_i$ is the $i'$th element of $\xb$. The smaller the value of function $S(\xb)$, the smoother the graph signal $\xb$.

\section{Generalized correntropy}
\label{sec:Gen_Corr}
Let $x$ and $y$ be two random variables with $F_{xy}(x,y)$ as the joint probability distribution function. Then, the definition of the correntropy is \cite{Chen16}
\begin{equation}
V(x,y)=\mathbb{E}\{\kappa(x,y)\}=\int{\kappa(x,y)dF_{xy}(x,y)},
\end{equation}
where $\mathbb{E}$ stands for the expectation operator, and $\kappa(x,y)$ is a kernel function. A usual kernel used in the correntropy is the Gaussian kernel:
\begin{equation}
\kappa(x,y)=\frac{1}{\sqrt{2\pi}\sigma}\exp(-\frac{e^2}{2\sigma^2})=\frac{1}{\sqrt{2\pi}\sigma}\exp(-\rho e^2),
\end{equation}
where $e=x-y$, $\sigma>0$ is the kernel width, and $\rho=1/2\sigma^2$. Here, we use generalized Gaussian density function (GGD) with zero-mean as the kernel function of the correntropy \cite{Chen16}:
\begin{equation}
\label{eq_kernel}
G_{\alpha,\beta}(e)=\frac{\alpha}{2\beta\Gamma(1/\alpha)}\exp(-|\frac{e}{\beta}|^\alpha)=z_{\alpha,\beta}\exp(-\rho|e|^\alpha)
\end{equation}
where $\alpha>0$ and $\beta>0$ are the shape and width (scale) parameters, respectively, $\Gamma(.)$ is the gamma function, $\rho=1/\beta^\alpha$, and $z_{\alpha,\beta}=\alpha/(2\beta\Gamma(1/\alpha))$ denotes the normalization constant. Using this GGD density function, and defining
\begin{equation}
V(x,y)=\mathbb{E}(G_{\alpha,\beta}(e))
\end{equation}
leads to generalized correntropy. As it is evident, substituting $\alpha=2$ leads to the correntropy with the Gaussian kernel. In practice, usually the joint pdf of $x$ and $y$ is not known, and only a limited number of samples $\{(x_i,y_i)\}_{i=1}^N$ are available. Thus, an estimate of generalized correntropy is obtained directly from samples as
\begin{equation}
\hat{V}(x,y)=\frac{1}{N}\sum_{i=1}^{N}{(G_{\alpha,\beta}(x_i-y_i))}=\frac{1}{N}\sum_{i=1}^{N}{(G_{\alpha,\beta}(e_i))}
\end{equation}
which we use in the current work, and we call it total generalized correntropy (TGC). In \cite{Chen16}, the generalized correntropy has been used as an estimation cost in signal recovery and a metric called the generalized correntropic loss (GC-loss) function is defined as
\begin{align}
J_{GC-loss}(x,y)&=G_{\alpha,\beta}(0)-V_{\alpha,\beta}(x,y)=z_{\alpha,\beta}-\mathbb{E}[G_{\alpha,\beta}(e_i)] \nonumber\\
&=z_{\alpha,\beta}\{1-\mathbb{E}[\exp(-\rho|e_i|^\alpha)]\}
\end{align}
where we have assumed the same kernel width for all samples, for simplicity. In the current work, we use an estimator of GC-loss function, called total GC-loss (TGC-loss) function as
\begin{align}
J_{TGC-loss}(x,y)&=G_{\alpha,\beta}(0)-\hat{V}_{\alpha,\beta}(x,y)=z_{\alpha,\beta}-\frac{1}{N}\sum_{i=1}^{N}{[G_{\alpha,\beta}(e_i)]} \nonumber\\
&=z_{\alpha,\beta}\{1-\frac{1}{N}\sum_{i=1}^{N}[\exp(-\rho|e_i|^\alpha)]\}
\end{align}

Thus, the unknown signal $x$ can be reconstructed from the measurements $y$ by minimizing the above TGC-loss function.

\section{The proposed adaptive algorithm}
\label{sec: proposed}
\subsection{TGC-loss based smoothed graph signal recovery}
Jointly applying the correntropy introduced in section \ref{sec:Gen_Corr} and smoothness of signal on graph presented in section \ref{sec:smooth}, the problem of reconstructing a smooth graph signal $\xb$ from sampled, noisy observations $\yb[n]$ can be defined as minimization of the following cost function
\begin{equation}
\label{eq:optimiz}
f(\xb)=\frac{1}{2}\xb^T\Lb\xb+\frac{\gamma}{2}z_{\alpha,\beta}\{1-\frac{1}{N}\sum_{i=1}^{N}[\exp(-\rho|e_i|^\alpha)]\},
\end{equation}
where $e_i=y_i-\Phi_{i,i} x_i$ denotes the $i$'th element of the error vector $\eb$, $\Phi_{i,i}$ is the $i'$th diagonal element of the $\Phi$, $\gamma$ is the regularization parameter, and $z_{\alpha,\beta}$ and $\rho$ are defined similar to (\ref{eq_kernel}). The first term in the above objective function encourages the smoothness of the reconstructed graph signal, and minimizing the second term is equivalent to maximizing the generalized correntropy, which maximizes the similarity between the observed and estimated graph signals.

\subsection{Adaptive algorithm}
\label{sec: prop}
In this section we exploit the stochastic gradient descent algorithm for the reconstruction of graph signal, which is a robust algorithm that has simple implementation and low computational cost, and, hence, is the most popular algorithm among the adaptive methods. The recovery process consists of estimating the smooth signal $\xb$ from sampled, noisy observations $\yb[n]$ (see \eqref{eq: model}) by solving the optimization problem (\ref{eq:optimiz}), which induces smoothness of the graph signal $\xb$ and maximizes the generalized correntropy. In order to distinguish between the original signal and the signal at any time instance, from now on, we show the original signal with $\xb^o$. The iterative stochastic gradient descent algorithm can be written as
\begin{equation}
\xb[n+1]=\xb[n]-\xi \nabla f(\xb[n])
\end{equation}
where $n$ is the time index, $\xi$ is the step-size parameter, and the gradient of function $f(\xb[n])$ is
\begin{equation}
\label{eq_nabla1}
\nabla f(\xb[n])=\Lb \xb[n]+\frac{\gamma\rho\alpha z_{\alpha,\beta}}{2}\Phi g(\eb[n])
\end{equation}
where
\begin{equation}
\label{eq_ge1}
g(\eb[n])=\frac{1}{N}\sum_{i=1}^N{(\exp(-\rho|e_i[n]|^\alpha)|e_i[n]|^{\alpha-1}sign(e_i[n])\hb_i)}
\end{equation}
and $\hb_i$ is an $N \times 1$ vector, whose all elements are zero except its $i'$th element, which is equal to one. Consequently, the gradient descent recursion of the algorithm is
\begin{equation}
\label{eq_LMS}
\xb[n+1]=\xb[n]-\xi\Lb \xb[n]-\frac{\xi\gamma\rho\alpha z_{\alpha,\beta}}{2}\Phi g(\eb[n])
\end{equation}

The proposed GC-GSR algorithm is given in Alg. \ref{Algorithm_1}.

\begin{algorithm}[!b]

\caption{The proposed GC-GSR Algorithm.}
\label{Algorithm_1}
\textbf{Input}   $\Lb,\quad \Phi,\quad \yb,\quad \rho,\quad\alpha, \quad\beta,\quad \xi,\quad\gamma $. \newline
\textbf{Initialize} $\xb[1]=\zerovec$, $\eb[1]=\zerovec$, $n=1$.

\begin{algorithmic}[1]
\REPEAT
\STATE
 {Find $\beta_i$ using the procedure in Alg. \ref{Algorithm_2} or use the given $\beta$};
\STATE
 {Obtain $g(\eb[n])$ by solving \eqref{eq_ge1}};
\STATE
{Compute $z_{\alpha, \beta}$ using \eqref{eq_kernel}};
\STATE
 {Update $\xb[n+1]$ using \eqref{eq_LMS}};
\STATE $n \leftarrow n+1$;
\UNTIL {a stopping criterion is reached.}
\end{algorithmic}
\end{algorithm}

\begin{Remark}
It seems that the proposed methods are rather direct applications of regularised empirical risk minimization, with total variation of graph signals as regularizer and a family of loss functions that is tailored to a specific noise distribution.
\end{Remark}

As a future work, one can put this work into context of the regularizated empirical risk minimization and its special case generalised total variation minimization, see \cite{Jung_book}-\cite{SarcheshmehPour21}. Moreover, one can put this work into context a line of work total variation minimisation methods Motivated by a clustered signal; see \cite{Jung18}- \cite{Jung20}, which, in contrast to our approach, use a non-smooth variant of total variation. This non-smooth variant
requires less sampled data points (somewhat similar to compressed sensing methods that use ell1 norm for regularization instead of $\ell 2$ norm squared).

While \cite{Jung18}- \cite{Jung20} provide error bounds without "any" assumption on the probability distribution of noise, \cite{Jung191} uses a Gaussian noise model to derive high-probability upper bounds on the estimation error.

\subsection{Kernel width estimation}
\label{sec:Kernel_est}
One of the parameters in the kernel function of the generalized correntropy is the scale parameter, which controls the width of the kernel function, and directly relates to the steady-state performance, convergence rate, and impulsive noise rejection. On the other hand, it is a free parameter that is often selected by the user. The value of the kernel width depend on the application properties, and, hence, determining an optimal value for it is important. In this paper, we have introduced a novel algorithm to address this issue, in which we assume a prior distribution on the (inverse of) the kernel width, and, then, learn its maximum a posteriori (MAP) estimate in an iterative procedure using the expectation maximization (EM) algorithm. The prior distribution must be chosen in such a way that it is a conjugate to the GGD distribution used for generalized correntropy, in order to obtain the posterior distribution in a closed form. Let $\theta_i=1/\beta_i$, then we assume a generalized Gamma distribution (G$\Gamma$D) with hyperparameters $a_0$, $d_0$, and $p_0$ over $\theta_i$ as
\begin{equation}
p(\theta_i)=G\Gamma D(a_0,d_0,p_0)=\frac{p_0/d_0^{a_0}}{\Gamma(d_0/p_0)}\theta_i^{d_0-1}\exp\left(-\left(\frac{\theta_i}{d_0}\right)^{a_0}\right)
\end{equation}

Thus, we can write
\begin{equation}
{\cal L}(\theta_i;x_i[n])\propto G_{\alpha,\beta}(e_i[n]).p(\theta_i)
\end{equation}
where ${\cal L}$ is the marginal likelihood of the observed data. Define $Q(\theta_i|\theta_i[n])$ as the expected value of the log-likelihood function of $\theta_i$, given $y_i[n]$ and the current estimate of the parameter ($\theta_i[n]$, where $n$ is the time index); then we can write the expectation step as

E-step:
\begin{align}
Q(\theta_i|\theta_i[n])&=\left<\ln({\cal L}(\theta_i;x_i[n],y_i[n]))\right>\nonumber\\
&=\left<\ln(G_{\alpha,\beta}(e_i[n]))+\ln(p(\theta_i))\right>+\mathrm{const.}\nonumber\\
&=\ln(\theta_i)-\theta_i^\alpha|e_i[n]|^\alpha+(d_0-1)\ln(\theta_i)-a_0^{-p_0}\theta_i^{p_0}+\mathrm{const.}
\end{align}

Substituting $p_0=\alpha$, we have
\begin{equation}
Q(\theta_i|\theta_i[n])=d_0\ln(\theta_i)-\left(|e_i[n]|^\alpha+a_0^{-\alpha}\right)\theta_i^\alpha
\end{equation}

In maximization step, we search for the parameter $\theta_i$ that maximize the above quantity:

M-step:
\begin{equation}
\theta_i[n+1]=\mathop{\mathrm{argmin}}_{\theta_i}Q(\theta_i|\theta_i[n])
\end{equation}
which leads to the following update equations for hyperparameters
\begin{align}
\label{eq_da}
&d[n+1]=d_0+1\nonumber\\
&a_i[n+1]=\left(\frac{1}{|e_i[n]|^\alpha+a_0^{-\alpha}}\right)^{\frac{1}{\alpha}}
\end{align}

\begin{algorithm}[!t]
\caption{Learning kernel width Algorithm.}
\label{Algorithm_2}
\textbf{Input}   $\eb, \alpha$. \newline
\textbf{Initialize} $d_0=1\times 10^{-6}$, $a_0=1\times 10^{-6}$.
\begin{algorithmic}[1]
\STATE Update $d$ and $a$ using \eqref{eq_da}
\STATE
 {Find $\theta_i$ using \eqref{eq_theta}}
\STATE
 {Set $\beta_i=1/\theta_i$}
\end{algorithmic}
\end{algorithm}

Finally, the update equation for the kernel width is
\begin{equation}
\label{eq_theta}
\theta_i[n+1]=a_i[n+1]\frac{\Gamma\left(\frac{d[n+1]+1}{\alpha}\right)}{\Gamma\left(\frac{d[n+1]}{\alpha}\right)}
\end{equation}

The resulting algorithm is summarized in the Alg. \ref{Algorithm_2}.

\section{Theoretical analysis}
\label{sec: theory}

\subsection{Convexity analysis}
In this subsection, we study the convexity of the cost function (\ref{eq:optimiz}). It has been proved in the literature that the first term of the cost function, which takes into account the smoothness of the signal, is convex. Therefore, we only discuss the second term, which is related to the correntropy, and we call it $\hat{f}$, i.e.,
\begin{equation}
\hat{f}(\eb[n])=z_{\alpha,\beta}\{1-\frac{1}{N}\sum_{i=1}^{N}[\exp(-\rho|e_i[n]|^\alpha)]\}
\end{equation}

\begin{theorem}
Let $\eb[n]=(e_1[n],...,e_N[n])^T$. Thus, we have

$i)$ if $0<\alpha<1$, then the cost function \eqref{eq:optimiz} is concave at any $\eb$ with $e_i\neq 0,\quad \forall i=1:N$;

$ii)$ if $\alpha>1$, then the cost function \eqref{eq:optimiz} is convex at any $\eb$ with $e_i\neq 0$ and $|e_i[n]|<(\frac{\alpha-1}{\rho\alpha})^\frac{1}{\alpha},\quad \forall i=1:N$;

$iii)$ if $\rho\longrightarrow 0$, then, for $0<\alpha<1$ the cost function \eqref{eq:optimiz} is concave at any $\eb$ with $e_i\neq 0$, and for $\alpha>1$, the cost function \eqref{eq:optimiz} is convex at any $\eb$ with $e_i\neq 0$.
\end{theorem}

\begin{proof}
Let $\Hb(\eb)$ denote the Hessian of the $\hat{f}$ with respect to $\eb[n]$. Thus, we have
\begin{equation*}
\label{eq_nabla2}
h_{i,j} =
\begin{cases}
\rho z_{\alpha,\beta}\alpha T(e_i)|e_i[n]|^{\alpha-2}\exp(-\rho|e_i[n]|^\alpha) & \text{if } i=j,\\
0 & \text{if } i\neq j
\end{cases}
\end{equation*}
where $h_{i,j}$ is the $(i,j)$'th element of the Hessian matrix $\Hb(\eb)$, and $T(v)=\alpha-1-\rho \alpha |v|^\alpha$. Thus, from the above equation, one can see that the sufficient condition for the convexity of $\hat{f}$ is
\begin{equation}
\alpha-1-\rho\alpha|e_i[n]|^\alpha>0,\forall i=1:N
\end{equation}
which leads to
\begin{equation}
\label{eq: cond1}
(\frac{\alpha-1}{\rho\alpha})^\frac{1}{\alpha}>|e_i[n]|,\forall i=1:N.
\end{equation}
Therefore, another sufficient condition for convexity of the cost function is
\begin{equation}
(\frac{\alpha-1}{\rho\alpha})^\frac{1}{\alpha}>|e|_{\mathrm{max}},
\end{equation}
where $|e|_{\mathrm{max}}$ is the maximum possible value of $|e_i[n]|$ which usually occurs at the initialization step of the algorithm. It can be see that if $0<\alpha<1$, then we have $\Hb(\eb)<\zerovec$ at any $\eb$ with $e_i\neq 0,\quad \forall i=1:N$. Moreover, if $\alpha>1$, then $\Hb(\eb)>\zerovec$ at any $\eb$ with $e_i\neq 0$ and $|e_i[n]|<(\frac{\alpha-1}{\rho\alpha})^\frac{1}{\alpha},\quad \forall i=1:N$. On the other hand, when $\rho\longrightarrow 0$, thus, for $0<\alpha<1$, we have $\Hb(\eb)<\zerovec$ at any $\eb$ with $e_i\neq 0$ and for $\alpha>1$, we have $\Hb(\eb)>\zerovec$ at any $\eb$ with $e_i\neq 0$.
\end{proof}

\subsection{Uniform stability performance}
In this subsection, we use uniform decreasing condition of the cost function to analyze the uniform stability performance of the proposed GSR algorithm, which is
\begin{equation}
\Delta f=f(\xb[n+1])-f(\xb[n])<0.
\end{equation}
The motivation for considering this uniform decrease in the cost function arises from a theorem in mathematical analysis which states that if a function of a sequence is decreasing and the function is lower bounded (for example be positive in our case), then the sequence converges \cite{MA02}.

According to the cost function defined by equation (\ref{eq:optimiz}), we have
\begin{align}
\Delta f=&\frac{1}{2}\xb^T[n+1]\Lb\xb[n+1]-\frac{1}{2}\xb^T[n]\Lb\xb_[n]\nonumber\\
&+\frac{\gamma z_{\alpha,\beta}}{2N}\sum_{i=1}^{N}[\exp(-\rho|e_i[n]|^\alpha)-\exp(-\rho|e_i[n+1]|^\alpha)]
\end{align}
which, using equation (\ref{eq_LMS}) can be approximated as
\begin{align}
\label{eq_nabla3}
\Delta f=&\frac{1}{2}[\left((\Ib-\xi\Lb)\xb[n]-\frac{\xi\rho\alpha\gamma z_{\alpha,\beta}}{2}\Phi g(\eb[n])\right)^T\Lb\nonumber\\
&\times\left((\Ib-\xi\Lb)\xb[n]-\frac{\xi\rho\alpha\gamma z_{\alpha,\beta}}{2}\Phi g(\eb[n])\right)-\xb^T[n]\Lb\xb[n]]\nonumber\\
&+\frac{\gamma z_{\alpha,\beta}}{2N}\sum_{i=1}^{N}{\exp(-\rho|e_i[n]|^\alpha)\left(1-\exp\left(-\rho\left(|e_i[n+1]|^\alpha-|e_i[n]|^\alpha\right)\right)\right)}\nonumber\\
\cong&\frac{\xi^2}{2}(\Lb\xb[n]+\frac{\rho\alpha\gamma z_{\alpha,\beta}}{2}\Phi g(\eb[n]))^T\Lb(\Lb\xb[n]+\frac{\rho\alpha\gamma z_{\alpha,\beta}}{2}\Phi g(\eb[n])) \nonumber\\
&-\xi(\Lb \xb[n]+\frac{\rho\alpha\gamma z_{\alpha,\beta}}{2}\Phi g(\eb[n]))^T\Lb \xb[n] \nonumber\\
&-\frac{\rho\gamma z_{\alpha,\beta}}{2N}\sum_{i=1}^{N}{\exp(-\rho|e_i[n]|^\alpha)\left(|e_i[n+1]|^\alpha-|e_i[n]|^\alpha\right)}
\end{align}

The last term at the summation of the above equation can be rewritten as
\begin{align}
|e_i[n+1]|^\alpha-|e_i[n]|^\alpha&=|e_i[n]+\Phi_i(\xb[n]-\xb[n+1])|^\alpha-|e_i[n]|^\alpha \nonumber\\
&=|e_i[n]+\Phi_i \tilde{\eb}[n]|^\alpha-|e_i[n]|^\alpha\nonumber\\
&\cong  |e_i[n]|^\alpha+\Delta P_i-|e_i[n]|^\alpha=\Delta P_i
\end{align}
where $\tilde{\eb}[n]=\xb[n]-\xb[n+1]$, and we have
\begin{align}
\Delta P_i=&|e_i[n]+\Phi_i \tilde{\eb}[n]|^\alpha-|e_i[n]|^\alpha\cong|e_i[n]|^\alpha\left(1+\alpha\frac{\Phi_i \tilde{\eb}[n]}{e_i[n]}\right)\nonumber\\
&-|e_i[n]|^\alpha\cong\alpha|e_i[n]|^{\alpha-1}\Phi_i \tilde{\eb}[n]
\end{align}

On the other hand, from (\ref{eq_LMS}) we have $\tilde{\eb}[n]=\xi(\Lb\xb[n]+\frac{\rho\alpha\gamma z_{\alpha,\beta}}{2}\Phi g(\eb[n]))$, and $\Phi_i=\hb_i^T\Phi^T$, thus, the summation term at the last line of (\ref{eq_nabla3}) can be rewritten as
\begin{align}
\frac{\rho\gamma z_{\alpha,\beta}}{2N}&\sum_{i=1}^{N}{\exp(-\rho|e_i[n]|^\alpha)\left(|e_i[n+1]|^\alpha-|e_i[n]|^\alpha\right)}\nonumber\\
=&\xi\frac{\alpha\rho\gamma z_{\alpha,\beta}}{2N}\sum_{i=1}^{N}\left({\exp(-\rho|e_i[n]|^\alpha)|e_i[n]|^{\alpha-1}\hb_i^T\Phi^T}\right)\nonumber\\
&\times(\Lb\xb[n]+\frac{\rho\alpha\gamma z_{\alpha,\beta}}{2}\Phi g(\eb[n]))\nonumber\\
=&\frac{\xi\alpha\rho\gamma z_{\alpha,\beta}}{2}(\Lb\xb[n]+\frac{\rho\alpha\gamma z_{\alpha,\beta}}{2}\Phi g(\eb[n]))^T\Phi g(\eb[n])
\end{align}

Finally, we have
\begin{align}
\Delta f=&\frac{\xi^2}{2}(\Lb\xb[n]+\frac{\rho\alpha\gamma z_{\alpha,\beta}}{2}\Phi g(\eb[n]))^T\Lb(\Lb\xb[n]+\frac{\rho\alpha\gamma z_{\alpha,\beta}}{2}\Phi g(\eb[n])) \nonumber\\
&-\xi(\Lb \xb[n]+\frac{\rho\alpha\gamma z_{\alpha,\beta}}{2}\Phi g(\eb[n]))^T\Lb\xb[n] \nonumber\\
&-\frac{\xi\alpha\rho\gamma z_{\alpha,\beta}}{2}(\Lb\xb[n]+\frac{\rho\alpha\gamma z_{\alpha,\beta}}{2}\Phi g(\eb[n]))^T\Phi g(\eb[n])\nonumber\\
=&\frac{\xi^2}{2}(\Lb\xb[n]+\frac{\rho\alpha\gamma z_{\alpha,\beta}}{2}\Phi g(\eb[n]))^T\Lb(\Lb\xb[n]+\frac{\rho\alpha\gamma z_{\alpha,\beta}}{2}\Phi g(\eb[n]))\nonumber\\
&-\xi(\Lb \xb[n]+\frac{\rho\alpha\gamma z_{\alpha,\beta}}{2}\Phi g(\eb[n]))^T(\Lb\xb[n]+\frac{\rho\alpha\gamma z_{\alpha,\beta}}{2}\Phi g(\eb[n]))
\end{align}

As mentioned above, the condition of uniform stability is that $\Delta f$ always be negative, which happens if we have
\begin{displaymath}
0<\xi<\frac{2\vb^T\vb}{\vb^T\Lb\vb},
\end{displaymath}
where $\vb\triangleq\Lb\xb[n]+\frac{\rho\alpha\gamma z_{\alpha,\beta}}{2}\Phi g(\eb[n])$. Therefore, using the Rayleigh quotient theorem, the sufficient condition for uniform stability is
\begin{equation}
0<\xi<\frac{2}{\lambda_{max}},
\end{equation}
where $\lambda_{max}$ is the maximum eigenvalue of the matrix $\Lb$.

\subsection{Mean performance}
In this subsection, we study the steady-state mean performance of the proposed algorithm. According to equation (\ref{eq_LMS}), we have
\begin{equation}
\mathbb{E}[\xb[n+1]]=(\Ib-\xi\Lb)\mathbb{E}[\xb[n]]-\xi'\Phi \mathbb{E}[g(\eb[n])]
\end{equation}
where $\xi'=\frac{\xi\rho\alpha\gamma z_{\alpha,\beta}}{2}$, and $g(\eb[n])$ is defined by (\ref{eq_ge1}). When the kernel parameter $\rho$ is small enough, we have
\begin{align}
g(\eb[n])\cong&\sum_{i=1}^N{\big((1-\rho|e_i[n]|^\alpha)|e_i[n]|^{\alpha-1}\mathrm{sign}(e_i[n])\hb_i\big)}\nonumber\\
=&\sum_{i=1}^N{\left(|e_i[n]|^{\alpha-1}\mathrm{sign}(e_i[n])\hb_i\right)}\nonumber\\
&-\sum_{i=1}^N{\left(\rho|e_i[n]|^{2\alpha-1}\mathrm{sign}(e_i[n])\hb_i\right)}
\end{align}

Thus, we have
\begin{align}
\mathbb{E}\left[g(\eb[n])\right]=&\mathbb{E}\left[\sum_{i=1}^N{\left(|e_i[n]|^{\alpha-1}\mathrm{sign}(e_i[n])\hb_i\right)}\right]\nonumber\\
&-\mathbb{E}\left[\sum_{i=1}^N{\left(\rho|e_i[n]|^{2\alpha-1}\mathrm{sign}(e_i[n])\hb_i\right)}\right]
\end{align}

By defining $\mathbb{E}[\sum_{i=1}^N{(\rho|e_i[n]|^m \mathrm{sign}(e_i[n])\hb_i)}]=\Ub_{n,m}$, the above equation can be rewritten as
\begin{equation}
\mathbb{E}[g(\eb[n])]=\Ub_{n,\alpha-1}-\rho\Ub_{n,2\alpha-1}
\end{equation}
and, thus
\begin{equation}
\mathbb{E}[\xb[n+1]]=(\Ib-\xi\Lb)\mathbb{E}[\xb[n]]-\xi'\Phi(\Ub_{n,\alpha-1}-\rho\Ub_{n,2\alpha-1})
\end{equation}

Let $\mathbb{E}[\xb[n]]=\Pb_n$, thus we have
\begin{equation}
\Pb_{n+1}=(\Ib-\xi\Lb)\Pb_n-\xi'\Phi(\Ub_{n,\alpha-1}-\rho\Ub_{n,2\alpha-1})=\Bb\Pb_n+\rb_n
\end{equation}
where $\Bb=\Ib-\xi\Lb$, and $\rb_n=-\xi'\Phi(\Ub_{n,\alpha-1}-\rho\Ub_{n,2\alpha-1})$. As can be seen, the above equation is a recursive equation, and we can write
\begin{equation}
\mathbb{E}[\xb[n]]=\Pb_n=\Bb^n\Pb_0+\Bb^{n-1}\rb_0+\Bb^{n-2}\rb_1+...+\Bb\rb_{n-2}+\rb_{n-1}
\end{equation}
\begin{equation}
\mathbb{E}[\xb[n+1]]=\Pb_{n+1}=\Bb^{n+1}\Pb_0+\Bb^n\rb_0+\Bb^{n-1}\rb_1+...+\Bb\rb_{n-1}+\rb_{n}
\end{equation}

As the algorithm reaches the steady-state, we have $\mathbb{E}[\xb[n]]\to\mathbb{E}[\xb[n+1]]$, which implies that
\begin{align}
&\Pb_{n+1}-\Pb_n\to 0\nonumber\\
&\implies(\Bb-I)(\Bb^n\Pb_0+\Bb^{n-1}\rb_0+\Bb^{n-2}\rb_1+...+\Bb\rb_{n-2}+\rb_{n-1})+\rb_n\to 0\nonumber\\
\end{align}
that yields
\begin{equation}
-\xi\Lb(\Bb^n\Pb_0+\Bb^{n-1}\rb_0+\Bb^{n-2}\rb_1+...+\Bb\rb_{n-2}+\rb_{n-1})+\rb_n\to 0
\end{equation}

Based on the previous assumption that the kernel parameter $\rho$ is small enough, we have $\Ub_{n,m}\to 0,\forall m$, and, then, $\rb_n \to 0$, thus, we can write
\begin{equation}
\Pb_{n+1}-\Pb_n\to-\xi\Lb(\Bb^n\Pb_0)
\end{equation}
which tend to zero if $\Bb^n=(\Ib-\xi\Lb)^n=0$, i.e.,
\begin{displaymath}
0<1-\xi\lambda_i<1,\forall i=1:N
\end{displaymath}
which yields
\begin{equation}
0<\xi<\frac{1}{\lambda_{max}}
\end{equation}
where where $\lambda_i$ and $\lambda_{max}$ are the eigenvalues and the maximum eigenvalue of the matrix $\Lb$, respectively.

\subsection{Computational complexity}
In this subsection, we provide a complexity analysis to the reconstruction process of the graph signal $\xb$ from sampled, noisy observations $\yb[n]$, where we have sampled and measured the signals of $M$ nodes (see \eqref{eq: model}).

Each iteration of the proposed GC-GSR algorithm for estimation the graph signal via (\ref{eq_LMS}) involves matrix-vector product and vector-vector summation, which consumes $\mathcal{O}(N^2)$ amount of computational complexity. Likewise, the LMS algorithm provided in \cite{Loren16} and the LMP algorithm in \cite{Nguy20} have $\mathcal{O}(M|\mathcal{F}|+N)$ order of complexity, where $|\mathcal{F}|$ denotes the number of non-zero elements in the Fourier transform of the signal $\xb$.

Table \ref{Table_1}, presents the detailed computational complexity comparison among the GC-GSR, LMS and LMP algorithms per iteration in terms of norm and arithmetic operations. It can be seen that the LMS algorithm, which benefits from the sparsity of the Fourier transform of the graph signal, has smaller computational complexity than the other algorithms. Moreover, the complexity of GC-GSR and LMP algorithms is a slightly larger than the LMS algorithm due to the inclusion of the smoothness term and the computational cost of the $\ell_p$-norm, respectively.
\begin{table*}[!t]
\vspace{-0.2cm}
\tiny
\centering
\caption{Computational Complexity of several iterative methods for graph signal recovery. }
\begin{tabular}{l||p{2cm}|p{1.6cm}|p{2.4cm}} \hline
& & &   \\[-1.5ex]
 $Algorithm$ & Multiplications & Additions & $\ell_p$ norm computation \\ \hline
  $GC-GSR$ & $N(M+N)+N$ & $2N$ & 0
 \\
  $LMS$ & $M|\mathcal{F}|+N $ & $2N$ & 0 \\
  $LMP$ & $M|\mathcal{F}|+N$ & $2N$  & $N$ \\ \hline
\end{tabular}
\label{Table_1}
\end{table*}

\section{Simulation Results}
\label{sec: simulation}
\subsection{Setup}
In this section, the performance of the proposed GC-GSR algorithm in adaptive graph signal recovery is investigated. In the first experiment, the effect of the parameters $\alpha$ and $\beta$ on the performance is examined, in which a synthetic graph signal is used. Secondly, to verify the robustness of the proposed algorithm, the performance of the proposed algorithm in presence of alpha-stable impulsive noise and GGD impulsive noise in comparison to some of the competing algorithms is evaluated. In the second experiment, a synthetic graph signal and a real-world graph signal of temperature data are used. The parameter setting of the noises is based on the following formula for GGD noise
\begin{equation}
\label{eq:GGD_noise}
v_i\sim\mathrm{GGD}(\nu,\eta)=\frac{\nu}{2\eta\Gamma(1/\nu)}\exp(-|\frac{v_i}{\eta}|^\nu),
\end{equation}
and the following formula for the characteristic function of a symmetric alpha-stable noise
\begin{equation}
\label{eq:alpha_noise}
\varphi(t)=\exp(j\mu t-\tau|t|^p),
\end{equation}
where $0<p<2$ is the stability parameter corresponding to the impulsiveness of the distribution, $\mu$ is the location parameter which equals the mean parameter for $1<p<2$ and the median parameter for $0<p<1$, and $\tau>0$ is the scale parameter.

\subsection{Synthetic Graphs and Signals}
For the synthetic graph signal, a similar graph used in \cite{Ioan19} is considered. The following ``seed matrix'' $\Pb_{0}$ is used similar to \cite{Ioan19} and \cite{Tork21}, which is:
\begin{displaymath}
\Pb_{0}=
\begin{bmatrix}
0.6 & 0.1 & 0.7 \\
0.3 & 0.1 & 0.5\\
0 & 1 & 0.1
\end{bmatrix}
\end{displaymath}
By defining $\Pb = \Pb_{0}\otimes \Pb_{0}\otimes \Pb_{0}\otimes \Pb_{0}$, where $\otimes$ denotes Kronecker product, a network with $N=81$ nodes are generated. Then, the weighted adjacency matrix $\Wb$ as $W_{ij}\sim \mathrm{Bernoulli}(P_{ij})$,  $\forall i,j$ can be considered. The bandlimited graph signals are generated by the model $\xb=\sum_{i=1}^{\omega} \gamma^{(i)} \cb^{(i)}$, where $\gamma^{(i)}\sim\cal{N}\mathrm{(0,1)}$, $\omega$ denote the bandwidth parameter, and $\{\cb^{(i)}\}^{\omega}_{i=1}$ are the eigenvectors related to the $\omega$ smallest eigenvalues of the Laplacian matrix. In the synthetic experiments, we set the bandwidth of the signal as $\omega=25$.

\begin{table*}[!t]
\vspace{-0.2cm}
\tiny
\centering
\caption{Performance of the proposed algorithms for learning kernel width in terms of NMSD (dB). }
\begin{tabular}{l||p{1cm}|p{1cm}|p{1cm}|p{1cm}|p{1cm}|p{1cm}|p{1cm}|p{1cm}|p{1cm}} \hline
& & & & &   \\[-1.5ex]
 $M$ & learned $\beta$ & $\beta=0.2$ & $\beta=1$ & $\beta=2$ & $\beta=5$ & $\beta=10$ & $\beta=15$ & $\beta=20$ & $\beta=30$\\ \hline \hline
  $30$ & $\bold{-17.2326}$ & -4.2489 & -4.7018  & -8.0131 & -10.7271  & -12.3560 & -12.1511 & -12.3560 & -10.0215
 \\
  $60$ & $\bold{-20.3081}$ & -8.5826 & -10.5900 & -12.7213 & -14.4098  & -15.0207  & -14.8912 & -14.7608 & -11.1601 \\
  $70$ & $\bold{-24.8029}$ & -12.3893 & -13.8726 & -15.1923 & -17.8230  & -18.0218  & -18.9023 & -18.5165 & -16.0290 \\
  $81$ & $\bold{-46.1243}$ & -17.9091 & -20.6512  & -23.7652 & -25.1345  & -31.5044 & -37.1109 & -37.3487 & -30.9173 \\ \hline
\end{tabular}
\label{Table_2}
\end{table*}

\begin{table*}[!t]
\vspace{-0.2cm}
\tiny
\centering
\caption{Performance of the proposed algorithms for different values of $\alpha$ in terms of NMSD (dB). }
\begin{tabular}{l||p{1cm}|p{1cm}|p{1cm}|p{1cm}|p{1cm}|p{1cm}|p{1cm}|p{1cm}|p{1cm}} \hline
& & & & &   \\[-1.5ex]
 $M$ & & $\alpha=1.05$ & $\alpha=1.1$ & $\alpha=1.15$ & $\alpha=1.3$ & $\alpha=1.5$ & $\alpha=2$ & $\alpha=4$ & $\alpha=8$ \\ \hline \hline
  \multirow{2}{*}{$30$} & learned $\beta$ & -10.8788 & -13.0450  &  -15.1080 & $\bold{-16.2118}$ & -15.2286 & -13.8890 & -11.0293 & -8.9711
 \\
 & fixed $\beta$ & -6.1768 & -7.1931  &  -10.2089 & $\bold{-12.2253}$ & -11.2277 & -9.9687 & -6.0304 & -4.5068 \\ \hline
  \multirow{2}{*}{$60$} & learned $\beta$ & -14.9274 & -16.7430 & -18.3231 & $\bold{-21.0109}$  & -20.1319  & -17.5127 & -14.0401 & -11.0435 \\
  & fixed $\beta$ & -9.6895 & -11.7841  &  -13.9161 & $\bold{-15.0284}$ & -14.3397 & -12.4285 & -9.0739 & -7.4073 \\ \hline
  \multirow{2}{*}{$81$} & learned $\beta$ & -30.2318 & -36.9910  & -40.0173 & $\bold{-46.7423}$  & -42.4150 & -34.9391 & -28.1069 & -20.7145 \\
  & fixed $\beta$ & -19.5318 & -25.7457  &  -30.0613 & $\bold{-38.0538}$ & -36.2876 & -29.6716 & -22.1926 & -418.8106 \\ \hline
\end{tabular}
\label{Table_3}
\end{table*}

The Normalized Mean-Square Deviation (NMSD) is used as the performance metric for evaluation of the GSR algorithms which is defined as
\begin{displaymath}
\mathrm{NMSD}[n]=\frac{||\xb^o-\xb[n]||^2}{||\xb^o||^2}
\end{displaymath}
The results are averaged on 100 independent Monte Carlo runs of the algorithm.

In the first experiment, we used a synthetic graph signal generated as stated above. For the impulsive noise, a GGD impulsive noise with parameters $\nu=1.3$ and $\eta=0.1$ is added to the synthetic graph signal to achieve $\mathrm{SNR}\sim 25\mathrm{dB}$. The $M$ samples of the graph with the total number of $N=81$ nodes are observed. Then, the proposed adaptive GSR algorithm is applied to the sampled noisy graph signal to reconstruct the graph signal.

\begin{figure*}[!t]
%\vspace{-1cm}
\centering
\subfloat[$M=50$]{\includegraphics[width=1.6in]{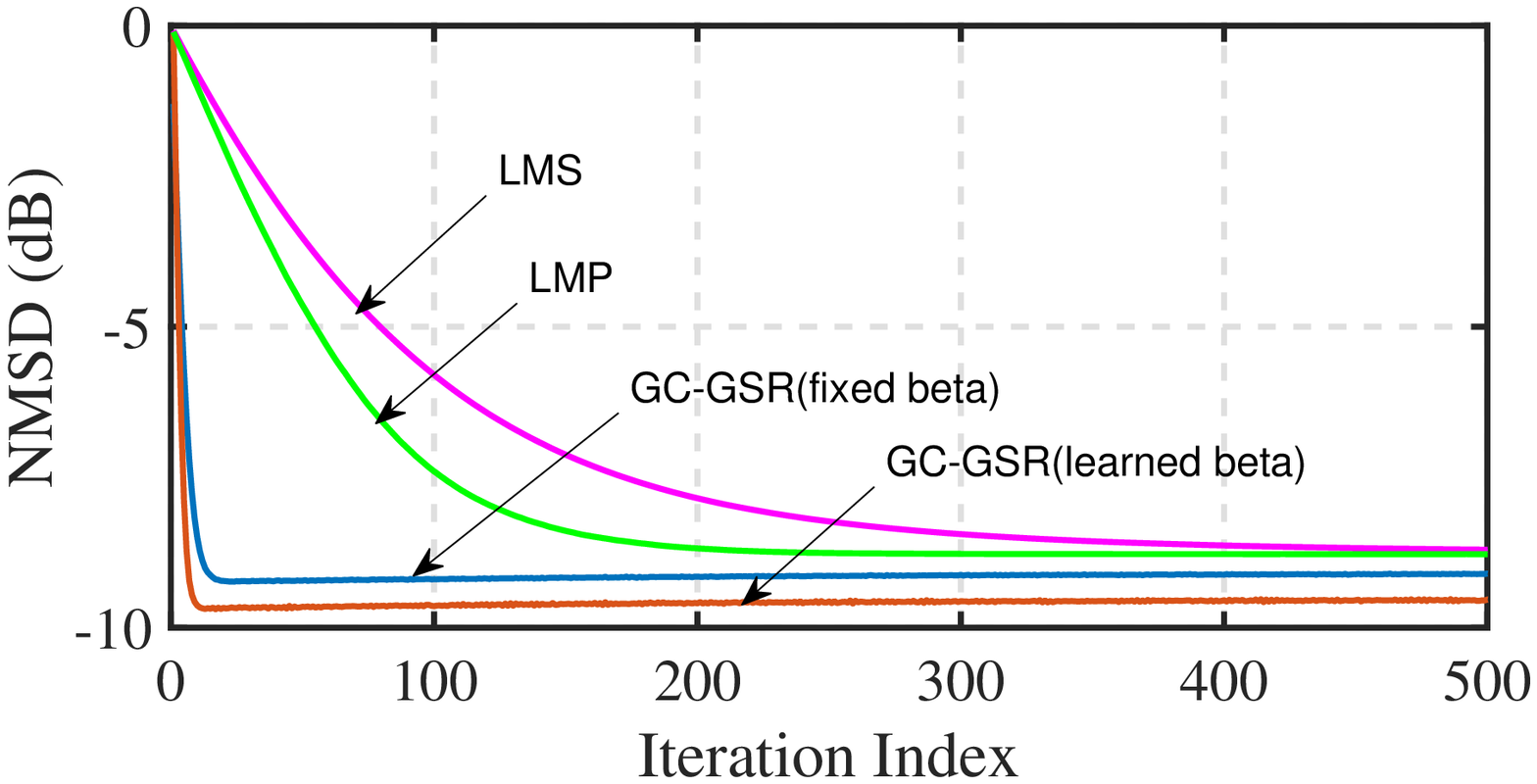}%
\label{fig1:a}}
\subfloat[$M=70$]{\includegraphics[width=1.6in]{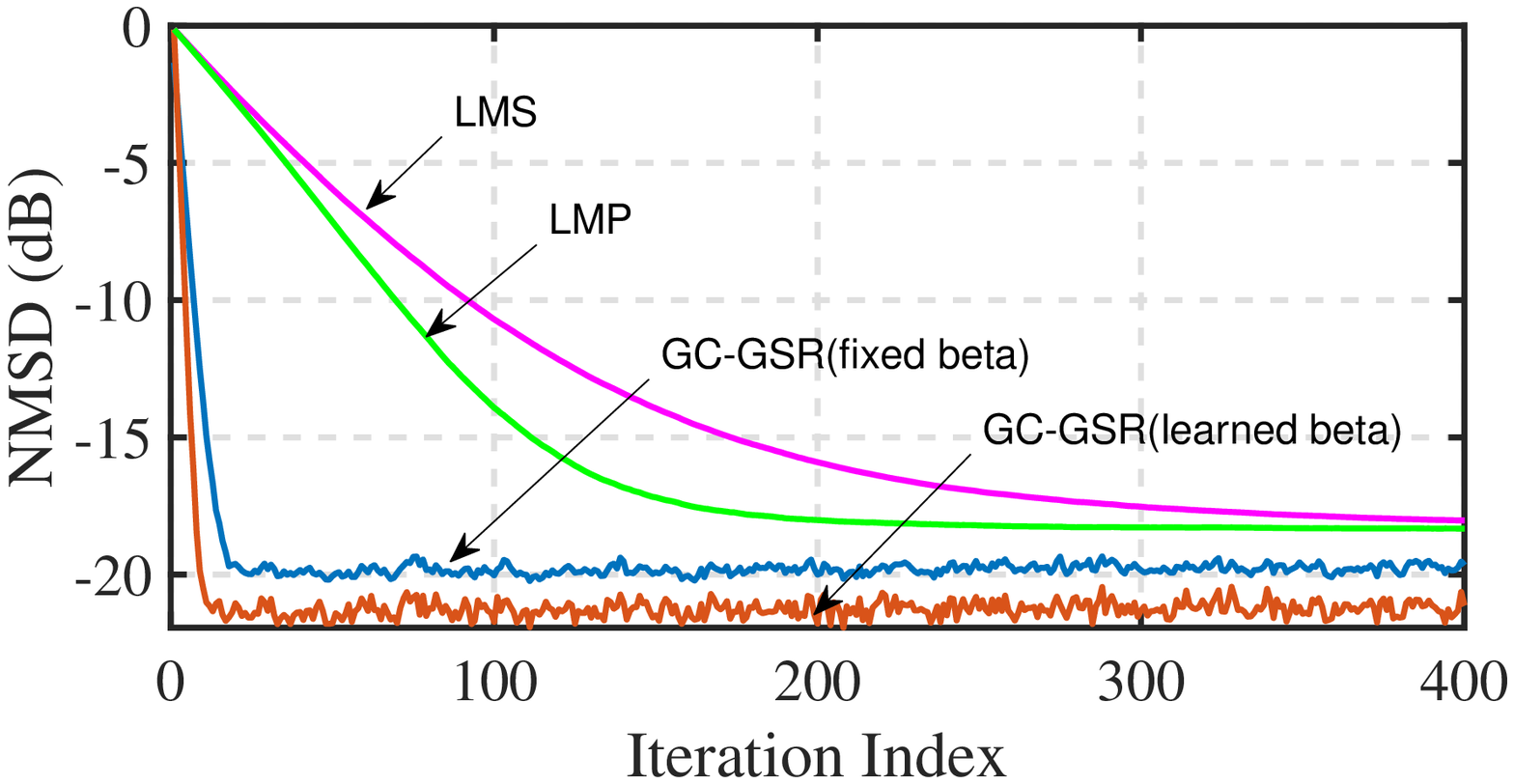}%
\label{fig1:b}}
\subfloat[$M=81$]{\includegraphics[width=1.6in]{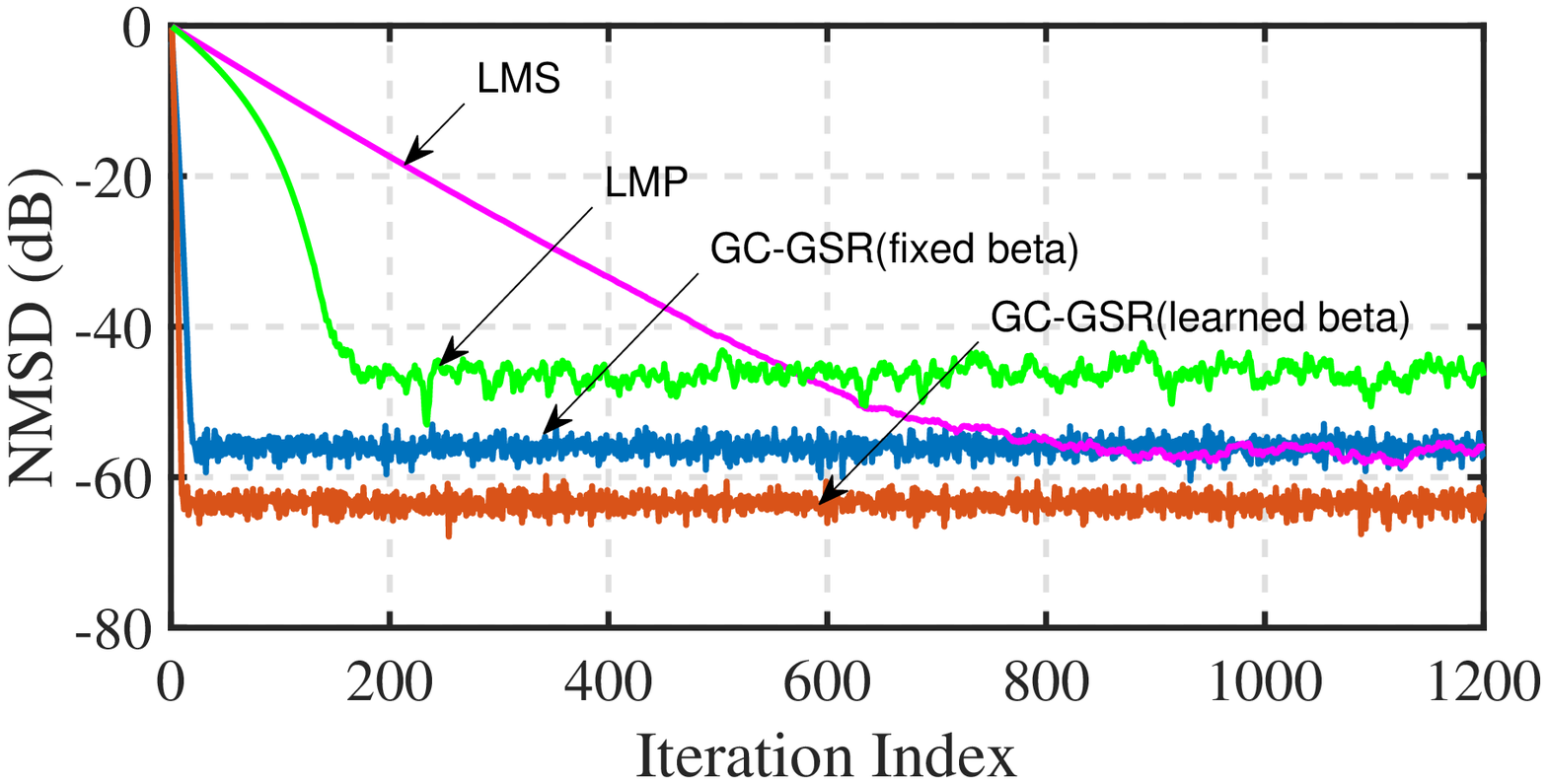}%
\label{fig1:c}}
\caption{Performance of algorithms for synthetic data in GGD noise.}
\label{fig1}
\end{figure*}

\begin{figure*}[!t]
%\vspace{-1cm}
\centering
\subfloat[$M=50$]{\includegraphics[width=1.6in]{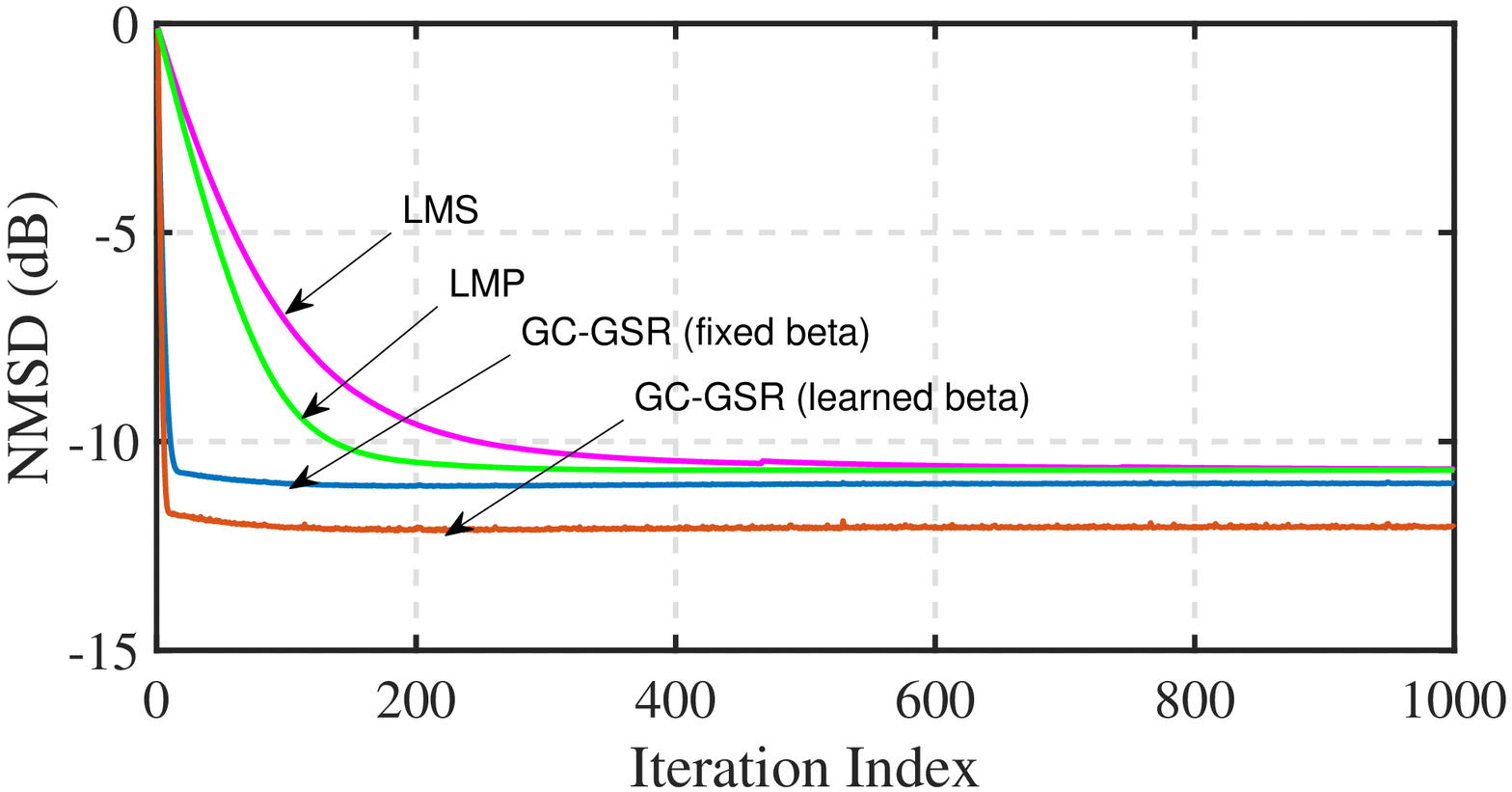}%
\label{fig2:a}}
\subfloat[$M=70$]{\includegraphics[width=1.6in]{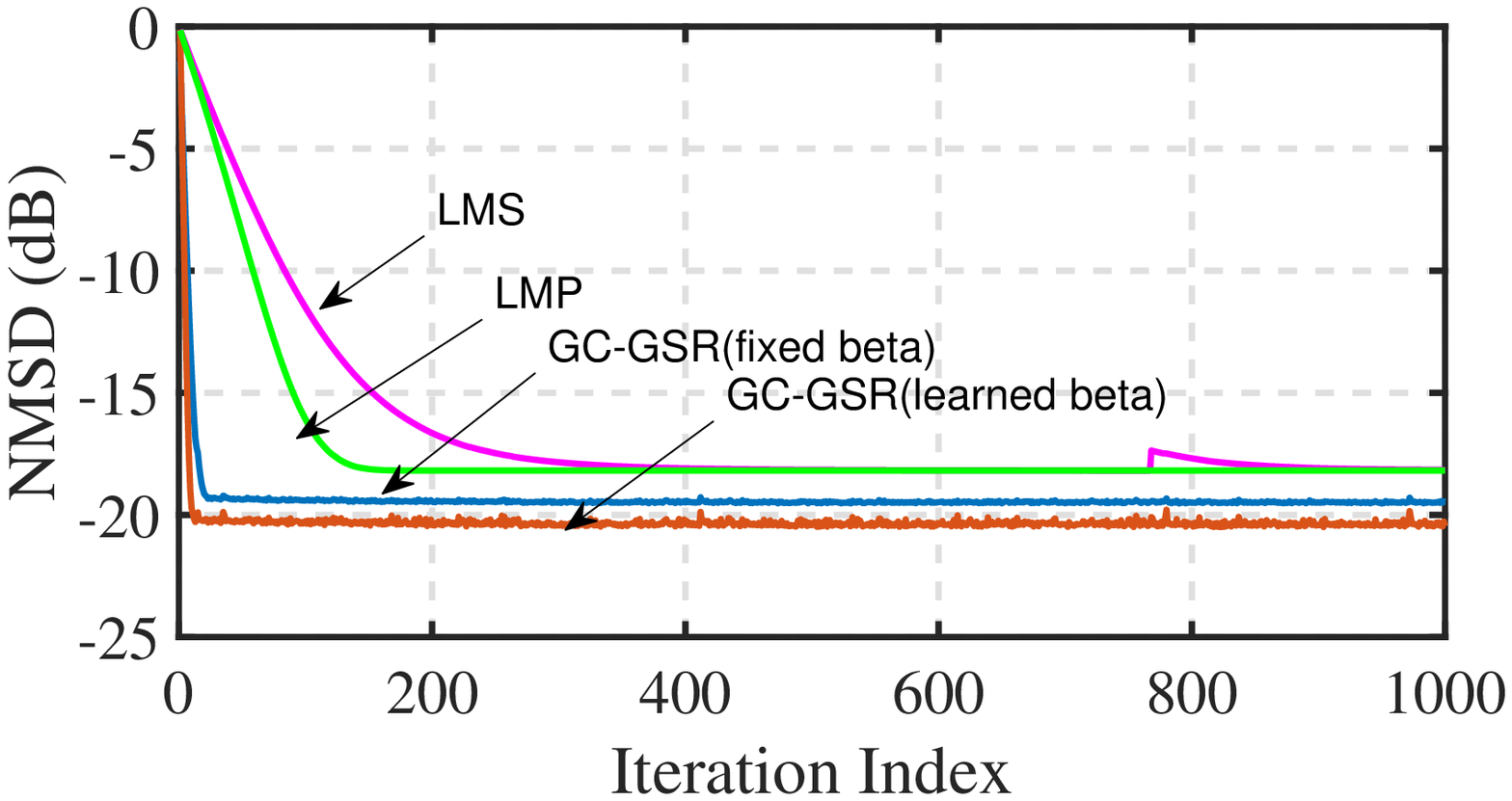}%
\label{fig2:b}}
\subfloat[$M=81$]{\includegraphics[width=1.6in]{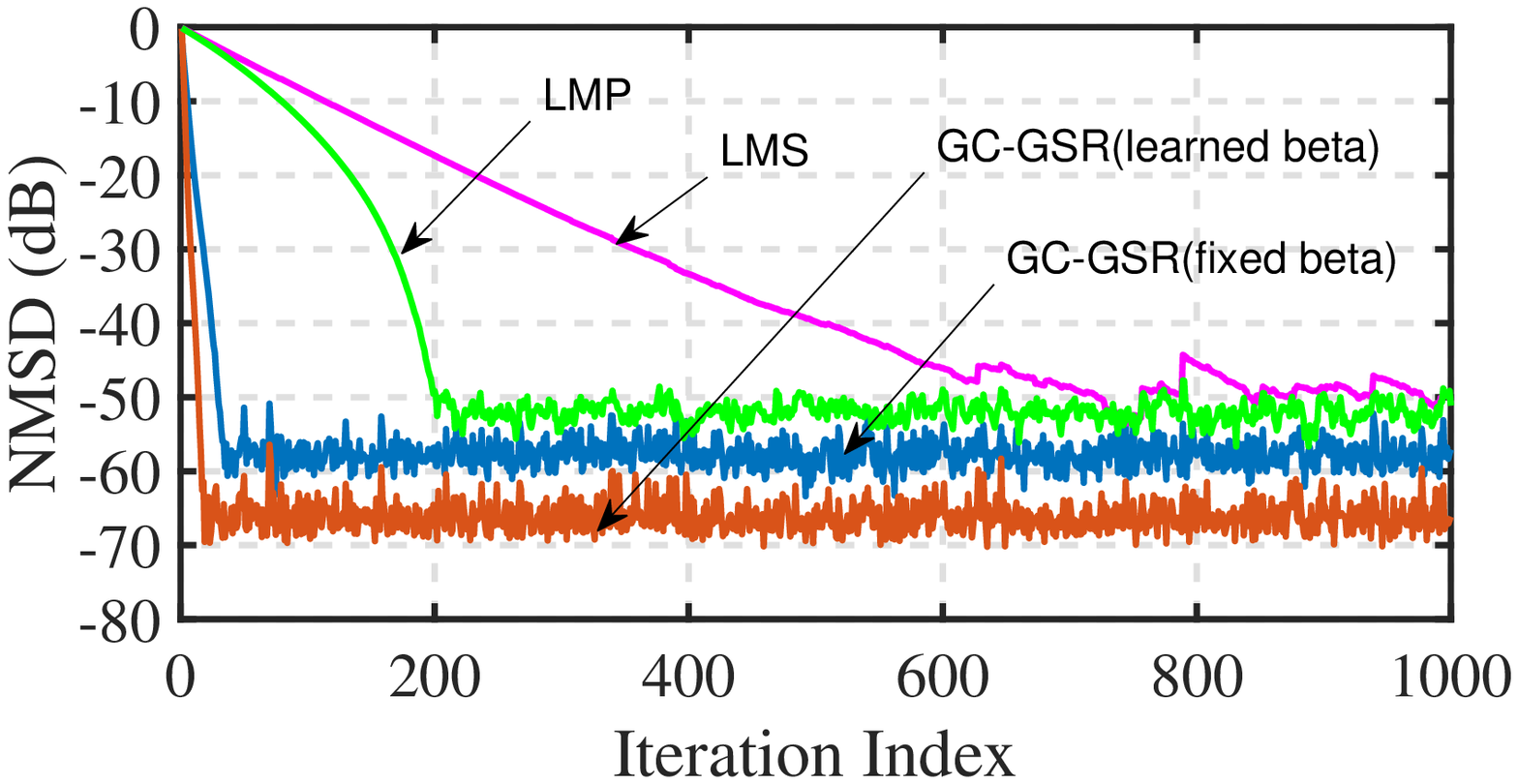}%
\label{fig2:c}}
\caption{Performance of algorithms for synthetic data in alpha-stable noise.}
\label{fig2}
\end{figure*}

\subsubsection{Performance of the algorithm with learning $\beta$}
In the first scenario, we evaluate the efficiency of the algorithm proposed for the kernel width estimation. To this end, we calculate the NMSD of the proposed algorithm for four values of $M(=30, 60, 70, 81)$, and in two conditions: i) when the parameter $\beta$ is learned using the algorithm presented in Section \ref{sec:Kernel_est}, and, ii) when we set a fix value for the parameter $\beta$. The results are presented in Table \ref{Table_2}. It can be seen that exploiting an algorithm for learning the kernel width parameter, which estimates different kernel widths for each element of graph signal based on an appropriate prior assumption for $\beta$, significantly reduces the respective reconstruction error of the graph signal.

\subsubsection{Effect of parameter $\alpha$}
The effect of parameter $\alpha$ (exponent) is explored in Table \ref{Table_3}, in which the NMSD of the proposed GC-GSR algorithm with different values of exponent $\alpha$, and for three values of $M(=30, 60, 81)$ is reported. It shows that regardless of the value of $\alpha$, the proposed algorithm reconstructs the graph signal with acceptable relative error.

\begin{figure*}[!t]
%\vspace{-1cm}
\centering
\subfloat[$SNR=10 dB$]{\includegraphics[width=1.6in]{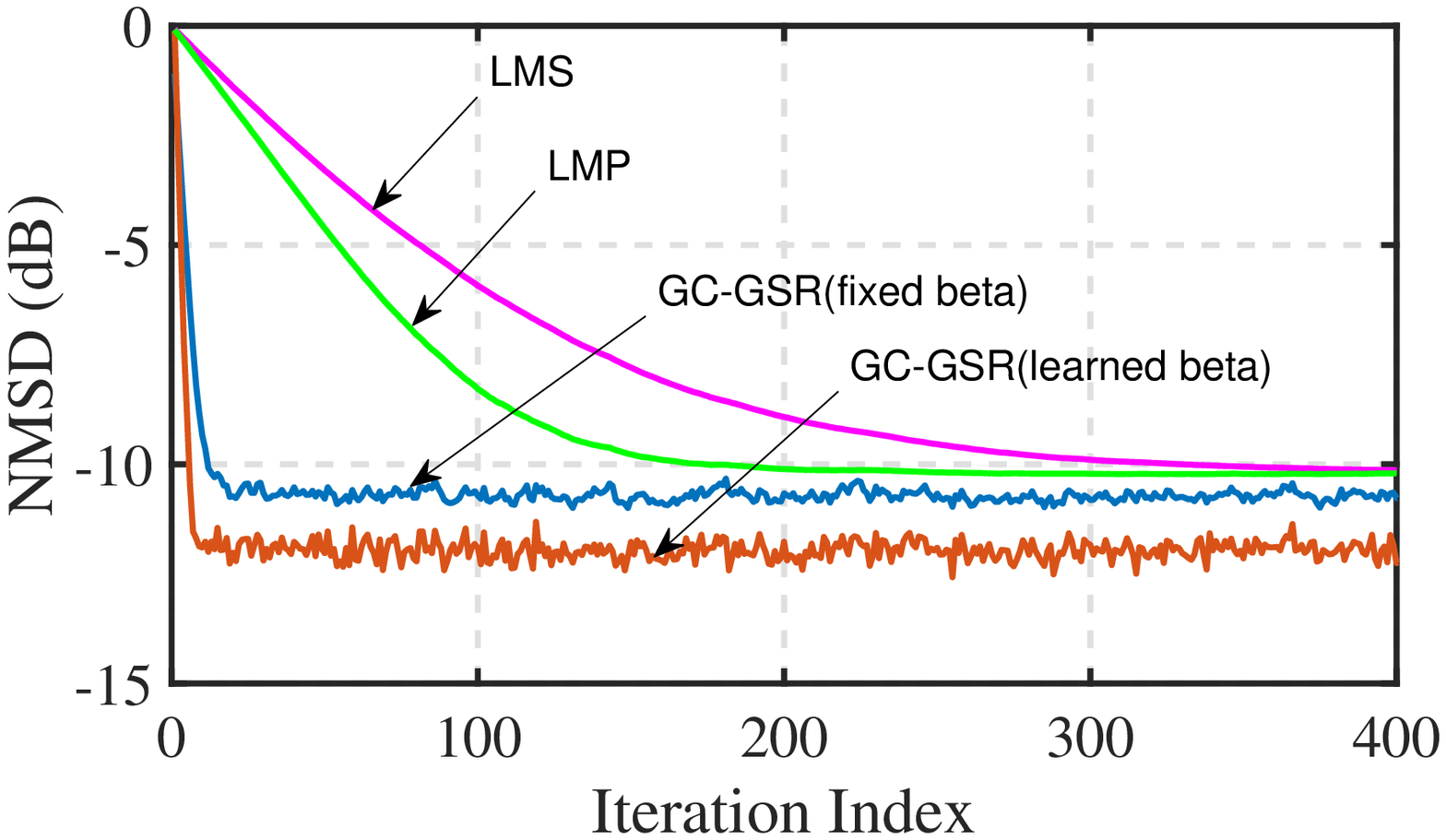}%
\label{fig3:a}}
\subfloat[$SNR=20 dB$]{\includegraphics[width=1.6in]{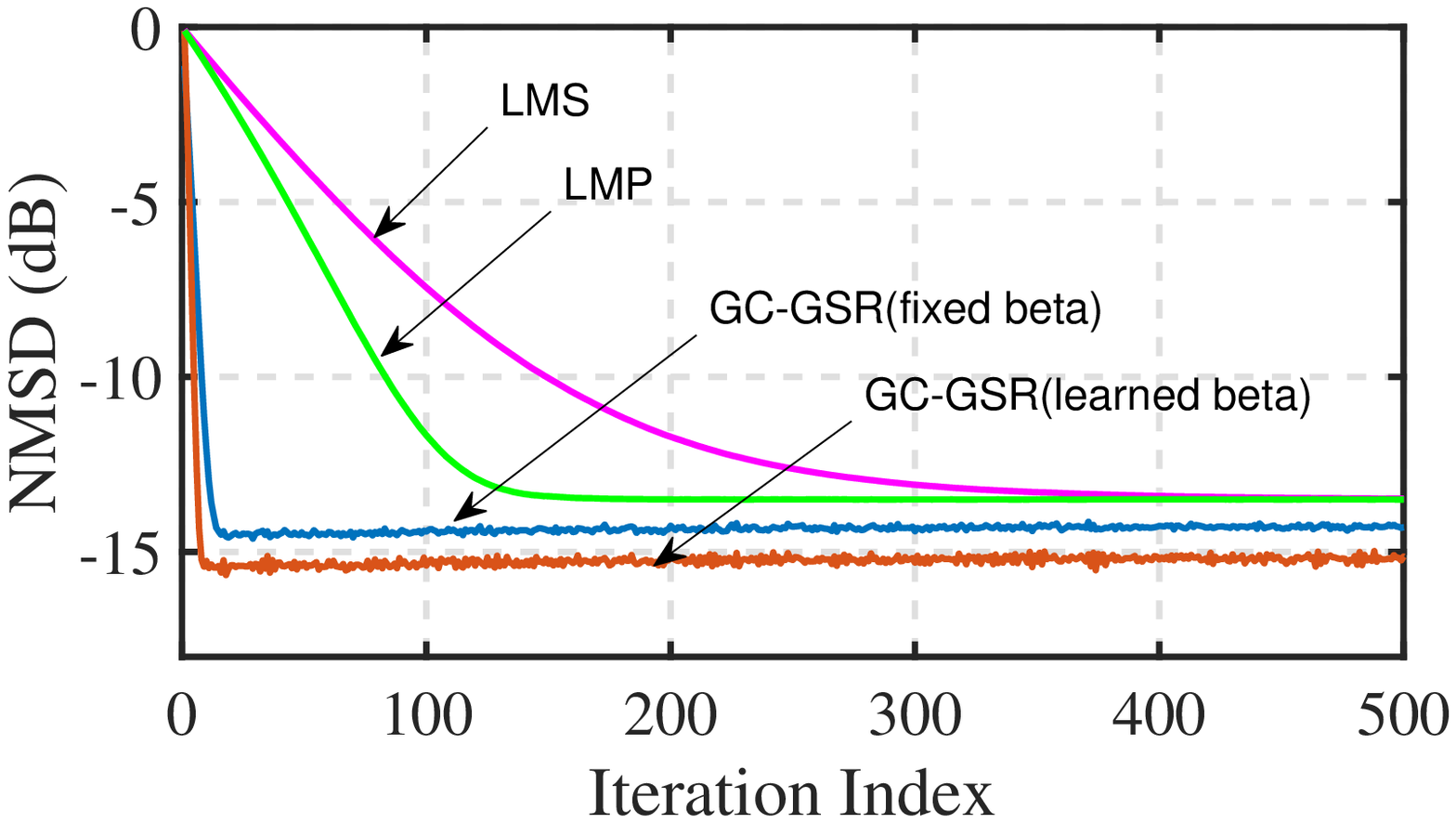}%
\label{fig3:b}}
\subfloat[$SNR=30 dB$]{\includegraphics[width=1.6in]{GGD_70.eps}%
\label{fig3:c}}
\caption{Performance of algorithms for synthetic data in GGD noise.}
\label{fig3}
\end{figure*}

\begin{figure*}[!t]
%\vspace{-1cm}
\centering
\subfloat[$\tau=0.05$]{\includegraphics[width=1.6in]{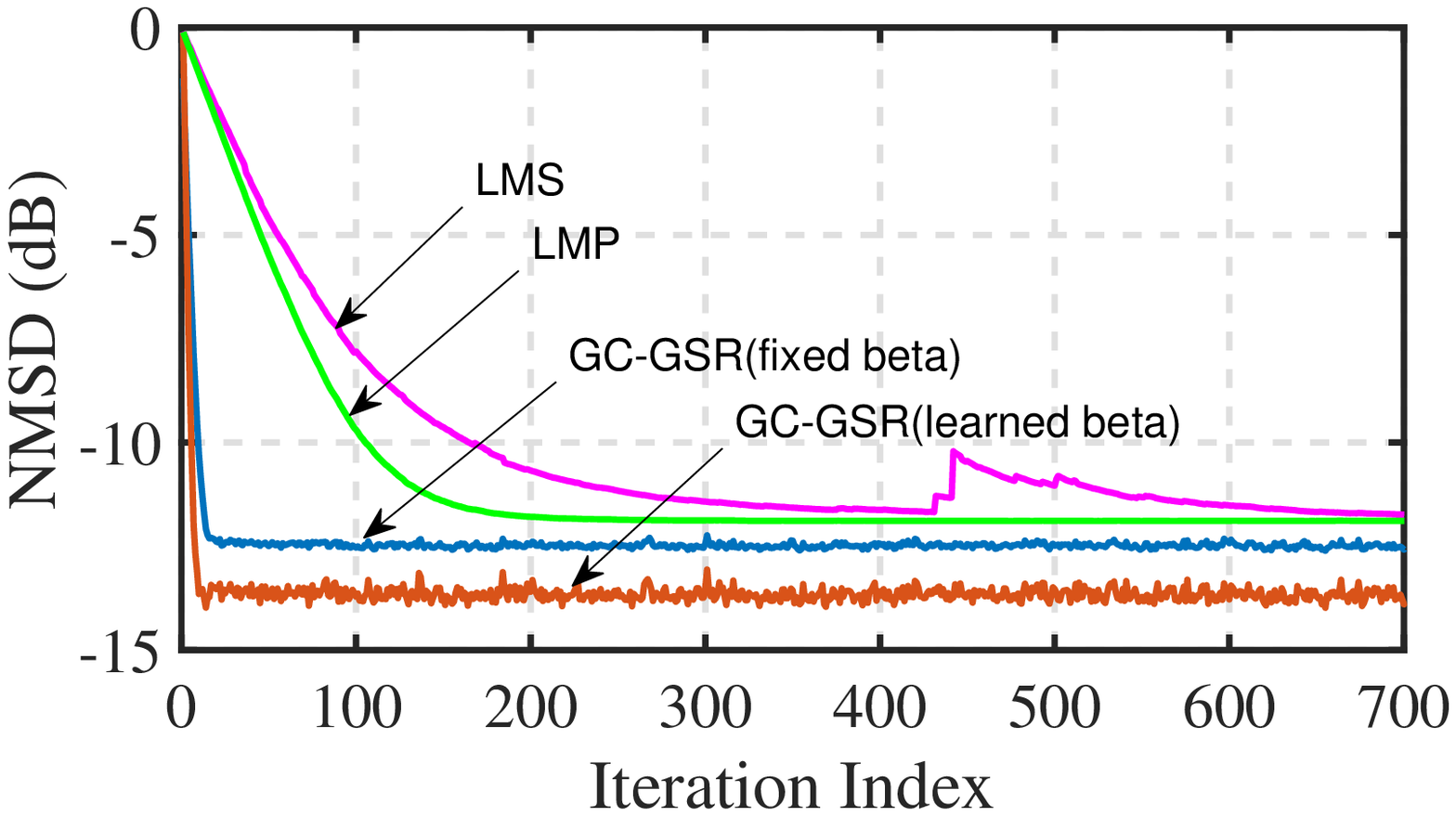}%
\label{fig4:a}}
\subfloat[$\tau=0.005$]{\includegraphics[width=1.6in]{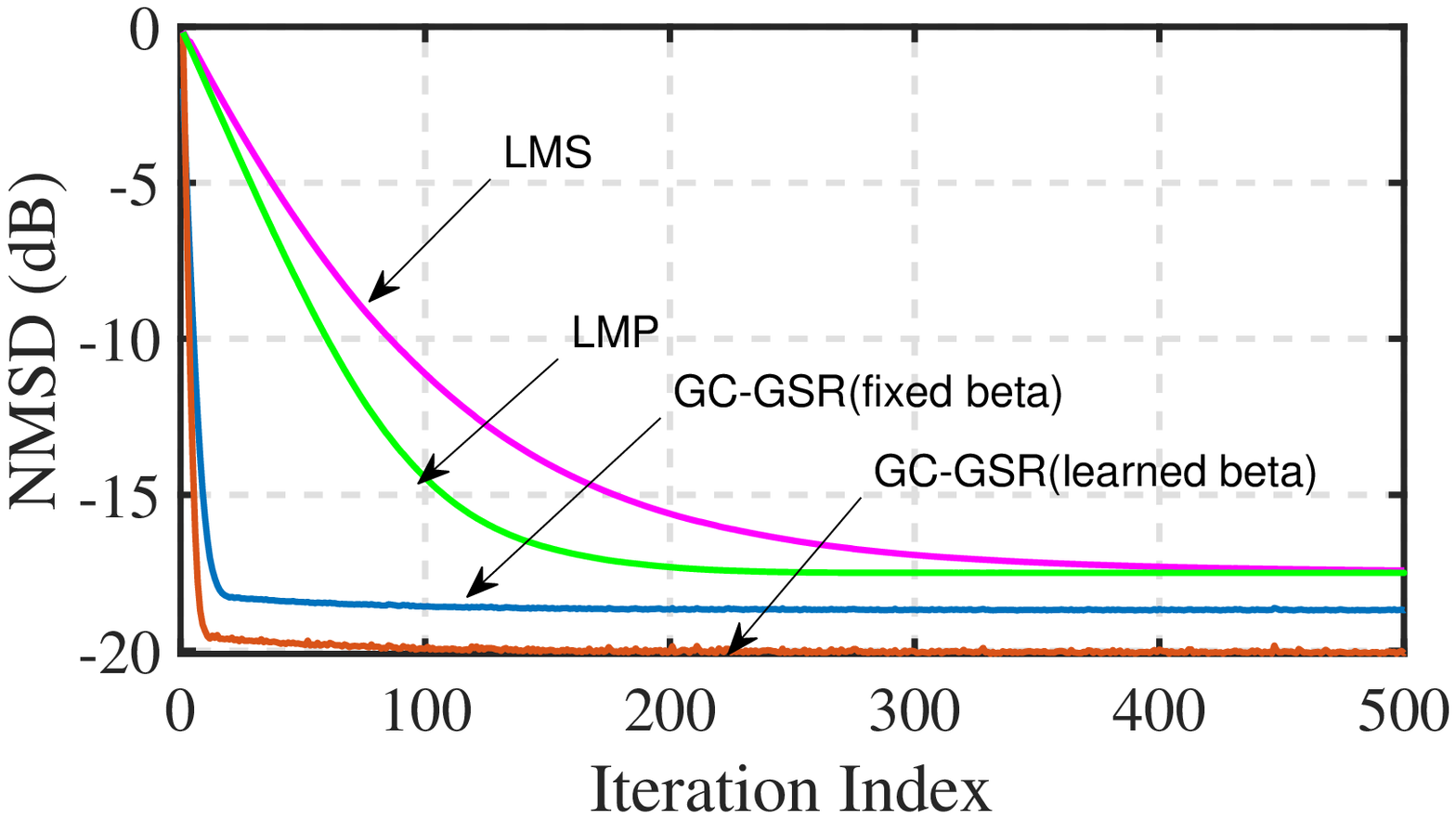}%
\label{fig4:b}}
\caption{Performance of algorithms for synthetic data in alpha-stable noise.}
\label{fig4}
\end{figure*}

\begin{figure*}[!t]
%\vspace{-1cm}
\centering
\subfloat[$M=30$]{\includegraphics[width=1.6in]{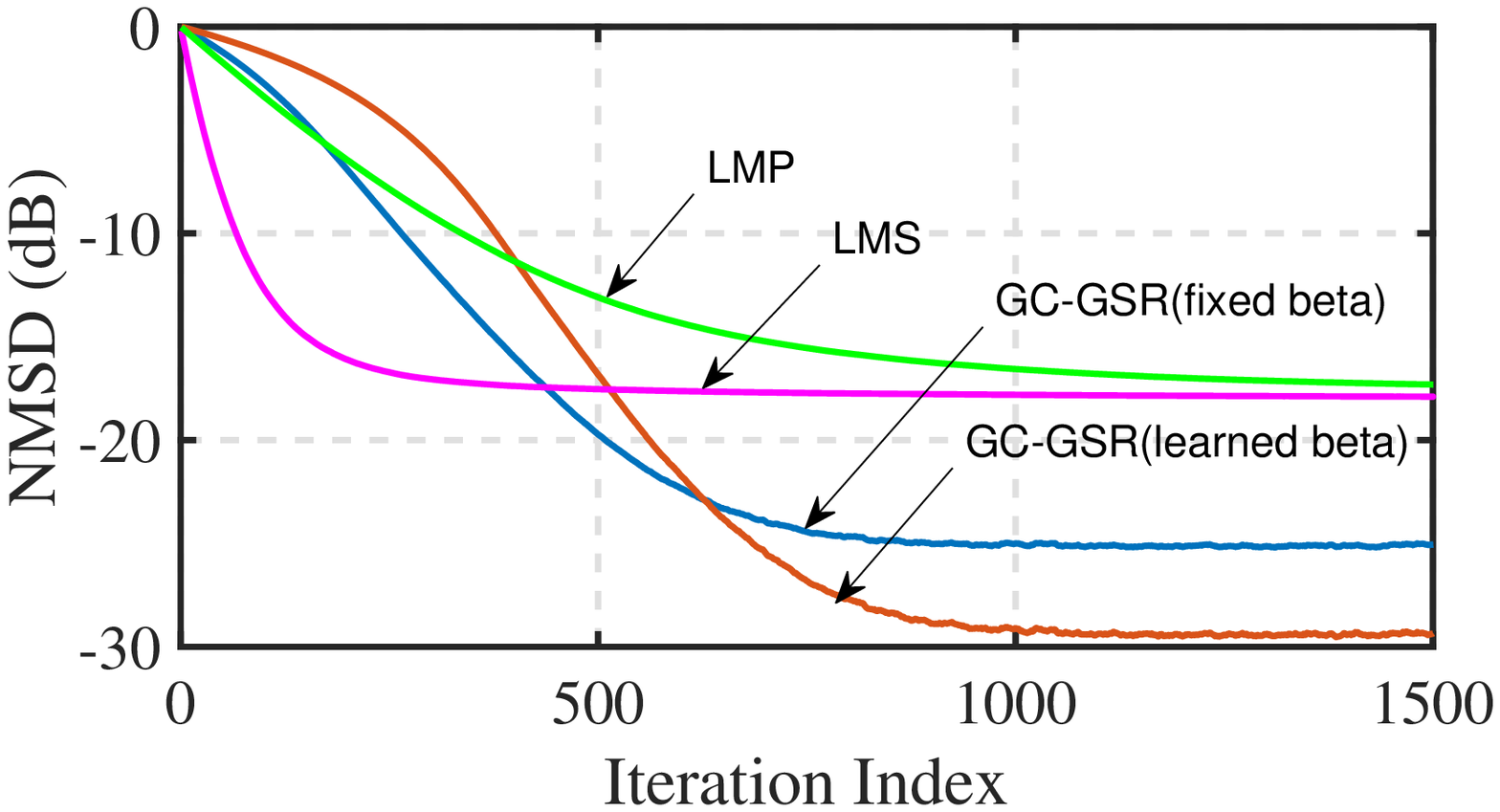}%
\label{fig5:a}}
\subfloat[$M=40$]{\includegraphics[width=1.6in]{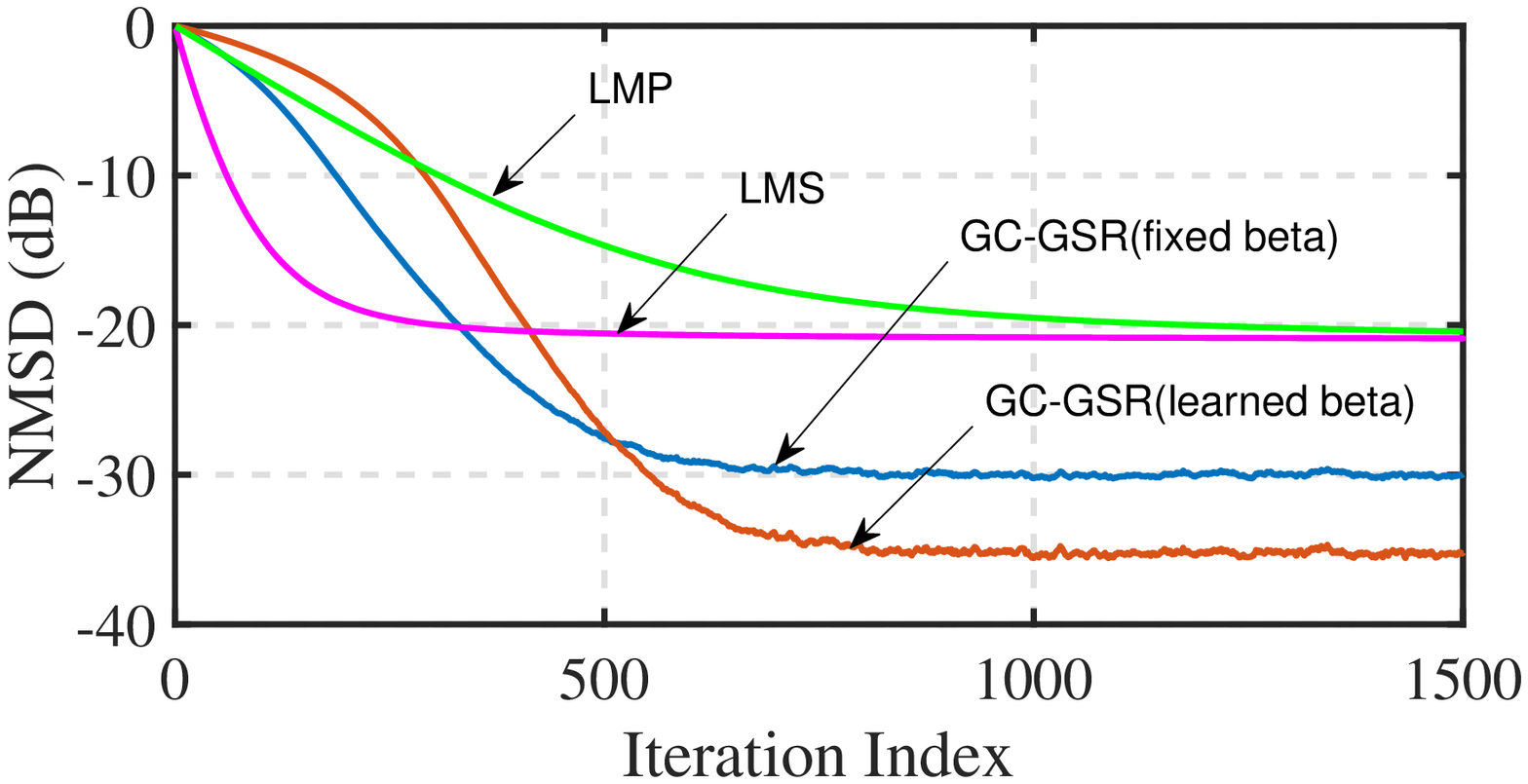}%
\label{fig5:b}}
\subfloat[$M=54$]{\includegraphics[width=1.6in]{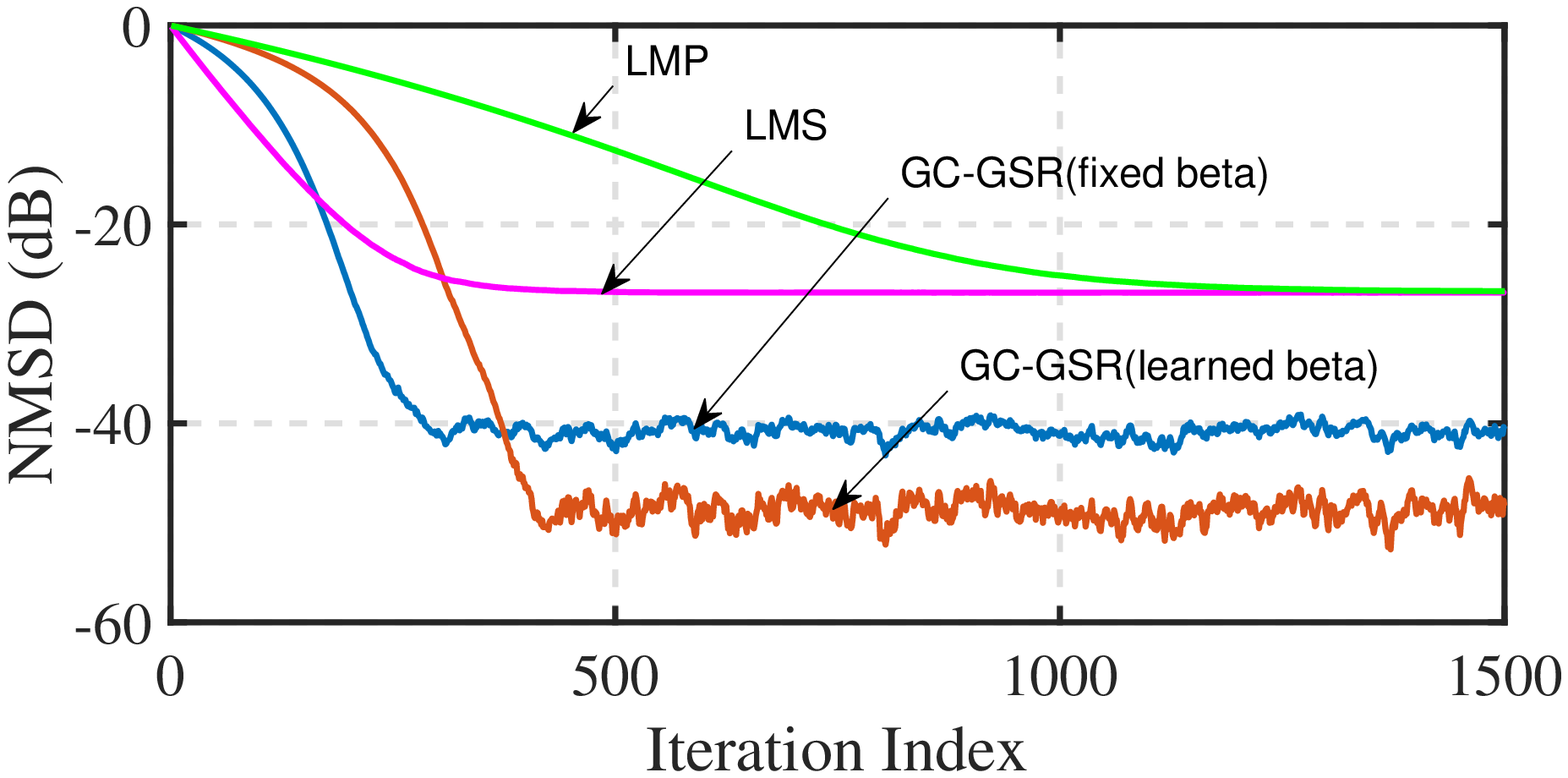}%
\label{fig5:c}}
\caption{Performance of algorithms for temperature data in GGD noise.}
\label{fig5}
\end{figure*}

\subsubsection{Comparison of recovery performance}
\begin{enumerate}[A]

\item \textsf{The reconstruction error with respect to sampling rate:} \newline
%The reconstruction error with respect to sampling rate: \\
In this experiment, the performance of the proposed GC-GSR is compared to other adaptive GSR algorithms which are LMS \cite{Loren16} and LMP \cite{Nguy20} for the case of synthetic graph signal as stated before, and in presence of impulsive noise. In this scenario, we considered two impulsive noises. The first impulsive noise is alpha-stable impulsive noise with parameters $\mu=0$, $p=1.3$ and $\tau=0.005$ (for further explanation of alpha-stable noise refer to \cite{Nguy20} and \cite{Zayy16}). In the experiments, the step-sizes are chosen such that all competing algorithms have almost the same initial convergence rate. The results of final NMSD versus the iteration index for three different values of graph samples ($M=50, 70, 81$) is depicted in figure \ref{fig1}. It shows that the proposed algorithm, which maximizes the generalized correntropy between the observations and the original graph signal has better result than LMS and LMP. In addition, although the LMP is designed for the alpha-stable noise cases, but it has poorer performance than the proposed GC-GSR algorithm. The reason is that the proposed algorithm, in addition to maximizing the correntropy criterion, uses GGD distribution as the kernel function, which is a good estimate for the alpha-stable distribution, and, thus, it can adopt to alpha-stable noise conditions.

The second impulsive noise is GGD impulsive noise with parameters $\eta=0.2$ and $\nu=1.3$. The result also are shown in figure \ref{fig2}. Again, the proposed method is the best method in terms of final NMSD. As can be seen from the experiments, for both cases of alpha-stable and GGD noises, the proposed algorithm outperforms the other methods, which shows the robustness of the proposed algorithm.

\item \textsf{The reconstruction error with respect to measurement noise} \newline
%The reconstruction error with respect to measurement noise \\
In this subsection, we demonstrate the performance of all the algorithms under different noise level. In this scenario, the number of sampled nodes is $M=70$.

First, we consider the observation noise to be GGD noise with $\nu=1.3$, and different noise levels result from selection of different values for $\eta$ in (\ref{eq:GGD_noise}). It can be seen from figure \ref{fig3} that the performance of all the algorithms improves with the increase of input SNR (decrease of $\eta$), and the proposed GC-GSR algorithm is always the best one.

Secondly, we assumed alpha-stable measurement noise with $p=1.3$ and $\mu=0$, and different noise levels result from different values of $\tau$ in (\ref{eq:alpha_noise}). Figure \ref{fig4} shows the NMSE performance of different algorithms for various values of input SNR.  It is observed that, in the presence of impulsive alpha-stable noise, the proposed GC-GSR algorithm, which was developed based on the maximum correntropy criterion, significantly outperforms the other algorithms.

Similar to the previous section, the experiments in this section verifies the robustness of the proposed algorithm. Moreover, it can be seen that the proposed algorithm adopts properly to the noise conditions.
\end{enumerate}

\subsection{Temperature Signal}
In the second experiment, we used the real-world temperature graph signal. In this case, the adaptive algorithms are applied to temperature estimation problem. We use the dataset downloaded from the Intel Berkeley Research lab (refer to \cite{Intel04}, \cite{Tork21}, and \cite{Tork22}), and the aim is to recover the unknown temperature values in a set of temperature values acquired from $N=54$ sensors. Similar to the first case, two impulsive noises are examined which are alpha-stable and GGD. The parameters $\mu$, $p$, and $\eta$ are the same as before. But, the other parameters are selected differently as $\tau=0.05$ and $\beta=0.5$. The NMSD versus the iteration index is depicted in figures \ref{fig5} and \ref{fig6} for three different numbers of selected nodes ($M=30, 40, 54$) and for alpha-stable noise and GGD noise, respectively. The figures show the superiority of the proposed algorithm in comparison to LMS and LMP algorithms.
\begin{figure*}[!t]
%\vspace{-1cm}
\centering
\subfloat[$M=30$]{\includegraphics[width=1.6in]{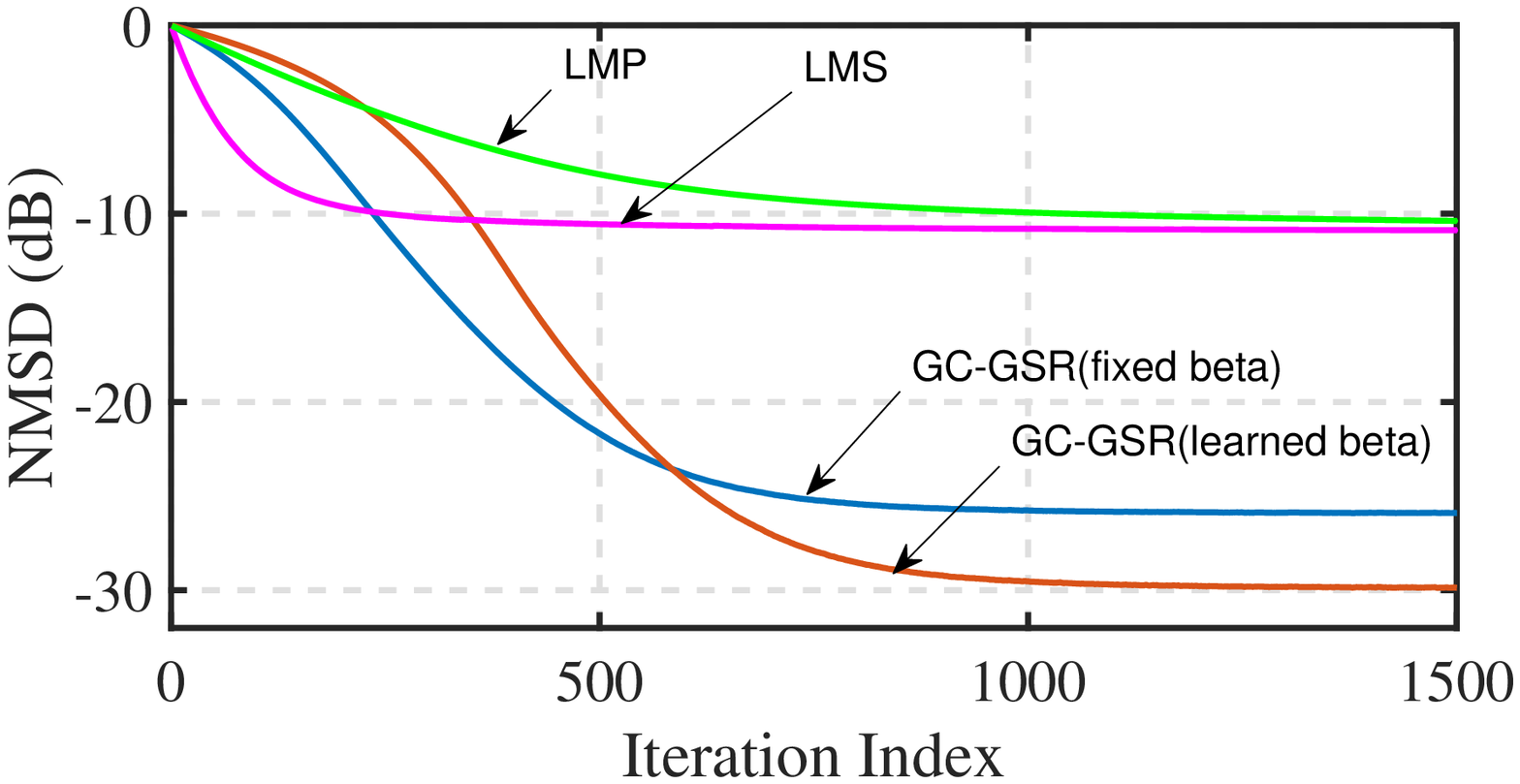}%
\label{fig6:a}}
\subfloat[$M=40$]{\includegraphics[width=1.6in]{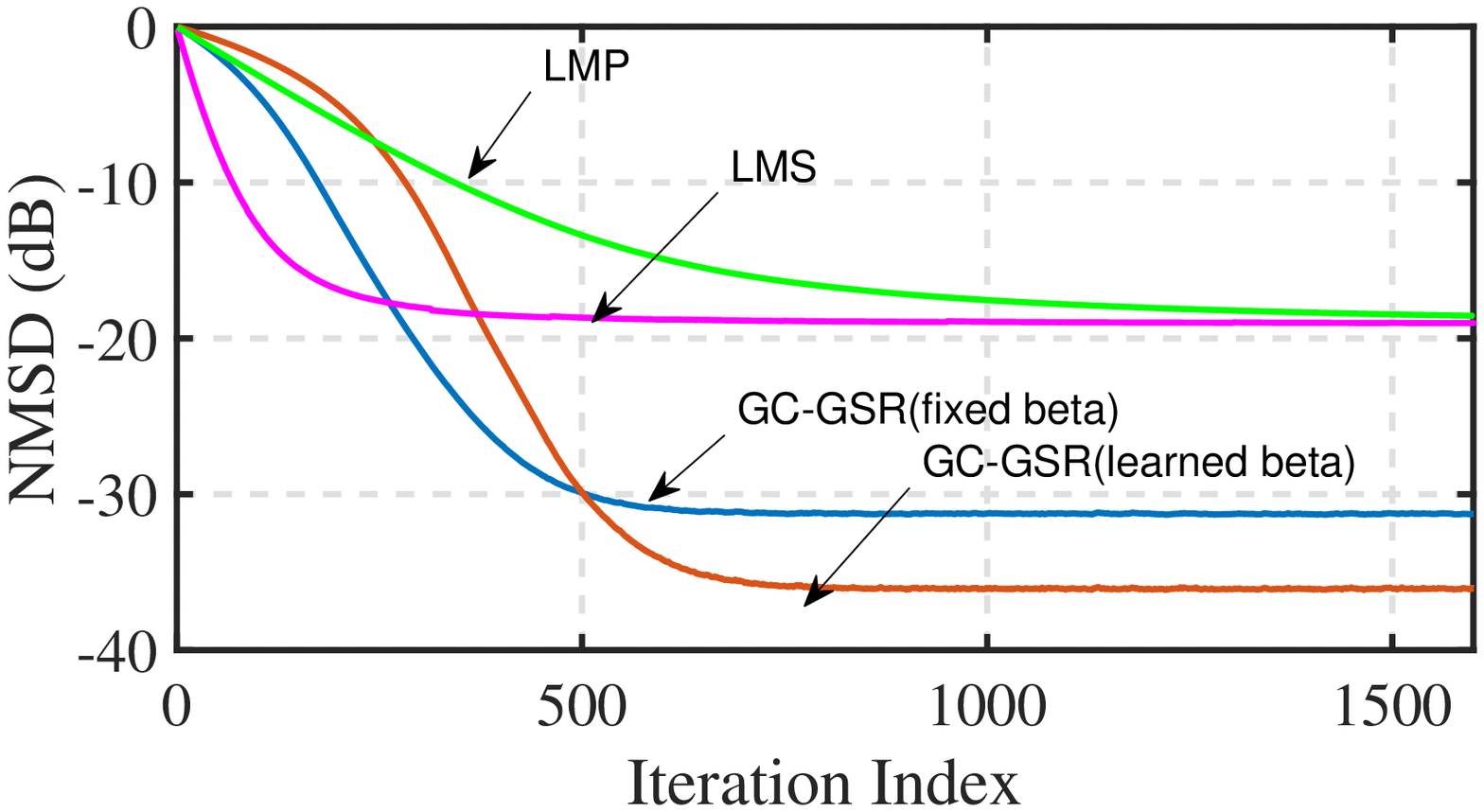}%
\label{fig6:b}}
\subfloat[$M=54$]{\includegraphics[width=1.6in]{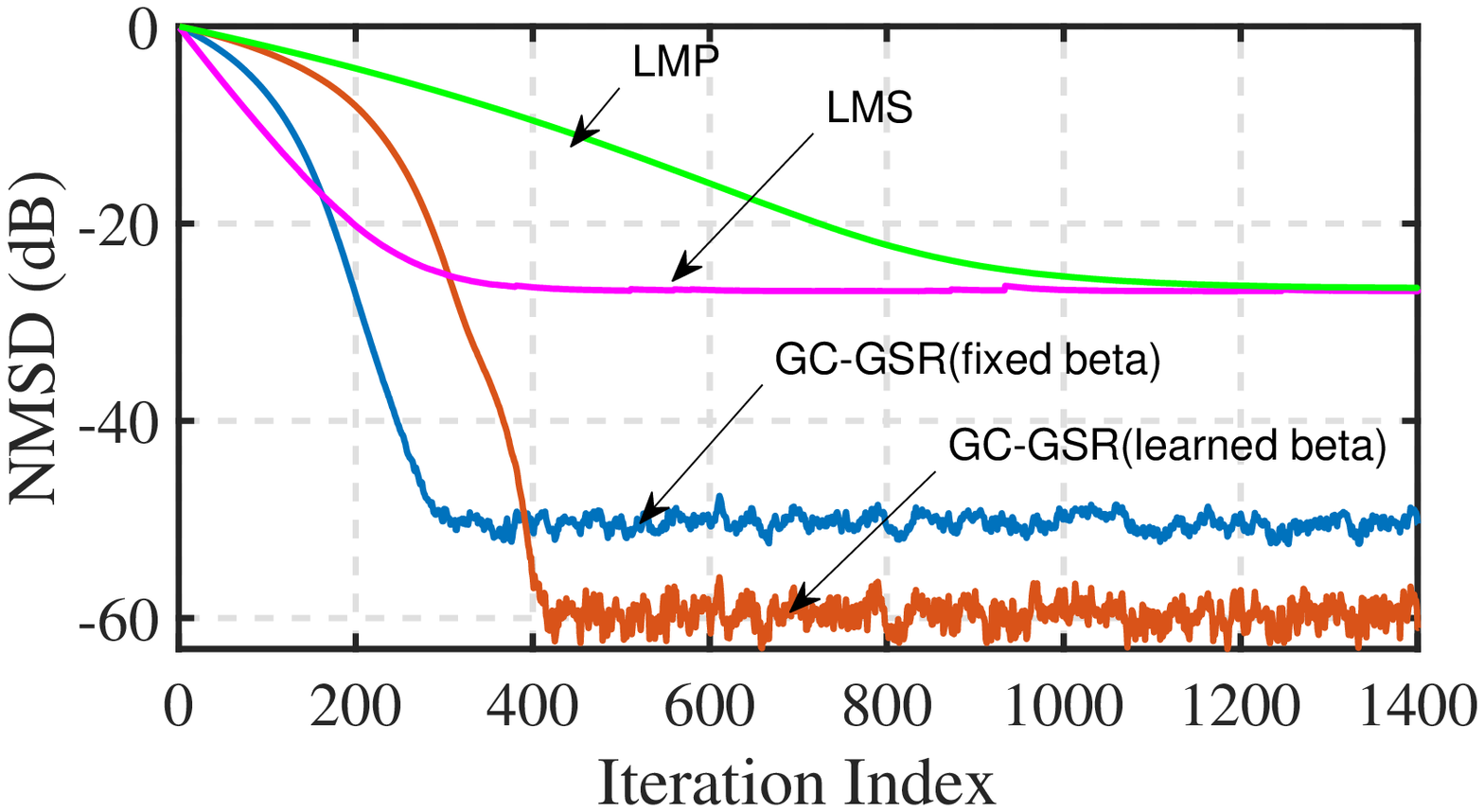}%
\label{fig6:c}}
\caption{Performance of algorithms for temperature data in alpha-stable noise.}
\label{fig6}
\end{figure*}

\section{Conclusion}
\label{sec: conclusion}
In this paper, a general correntropy-based adaptive graph signal recovery algorithm is devised which is suitable for robust smoothed graph signal recovery in the presence of impulsive noise. The proposed adaptive algorithm suggests to use a cost function which is a combination of graph smoothness term and a generalized correntroy term. The generalized correntropy term has two parameter of exponent and kernel width. A kernel width learning mechanism is developed in the algorithm which is based on Bayesian inference of $\beta$. Moreover, a convexity analysis of the cost function, the uniform stability, the mean convergence of the algorithm, and the complexity analysis is added in the section of theoretical analysis. In the simulation results, the benefit of learning $\beta$ is firstly shown. Then, in the experiments presented in the simulation results section, it was deduced that the proposed GC-GSR algorithm has the best result in comparison to other competing algorithms which are LMS and LMP. Moreover, the experiments verify the robustness of the proposed algorithm.

\section{Acknowledgement}
This work was supported by the Iran National Science Foundation (INSF) (grant number
4005022).% <-this %

\end{document}